\documentclass[lettersize,journal]{IEEEtran}
\usepackage{amsmath,amsfonts}
\usepackage{algorithmic}
\usepackage{array}
\usepackage{xargs}                  
\usepackage{textcomp}
\usepackage{stfloats}
\usepackage{url}
\usepackage{verbatim}
\usepackage{graphicx}
\graphicspath{{./Figures/}}
 \usepackage{latexsym}
\usepackage[table,xcdraw]{xcolor}
\usepackage{tabularx}
\usepackage{amssymb, amsthm}
\usepackage{placeins}
\usepackage[font=footnotesize,labelfont=bf]{caption}
\usepackage{subcaption}
\usepackage{mathtools}
\usepackage{multirow}
\usepackage{cite, enumerate}
\usepackage[inline,shortlabels]{enumitem}
\usepackage{dsfont}
\usepackage{siunitx} 
\usepackage{accents}
\usepackage[normalem]{ulem}
\usepackage{tikz}
\usepackage{circuitikz}

\def\BibTeX{{\rm B\kern-.05em{\sc i\kern-.025em b}\kern-.08em
		T\kern-.1667em\lower.7ex\hbox{E}\kern-.125emX}}
\usepackage{balance}

\usepackage{acronym}
\usepackage[shortcuts]{glossaries}
\newacronym{mimo}{mMIMO}{massive multiple-input multiple-output}
\newacronym{mse}{MSE}{mean-squared-error}
\newacronym{nmse}{NMSE}{normalized MSE}
\newacronym{mmse}{MMSE}{minimum mean squared error}
\newacronym{csi}{CSI}{channel state information}
\newacronym{awgn}{AWGN}{additive white Gaussian noise}
\newacronym{cf}{CF}{cell-free}
\newacronym{tdd}{TDD}{time-division duplex}
\newacronym{dtdd}{DTDD}{dynamic time-division duplex}
\newacronym{fd}{FD}{full-duplex}
\newacronym{cpu}{CPU}{central processing unit}
\newacronym{roc}{ROC}{receiver operating characteristic}
\newacronym{snr}{SNR}{signal-to-noise ratio}
\newacronym{sinr}{SINR}{signal-to-inteference-plus-noise ratio}
\newacronym{scnr}{SCNR}{signal-to-clutter-plus-noise ratio}
\newacronym{cn}{CSCN}{circularly symmetric complex normal}
\newacronym{se}{SE}{spectral efficiency}
\newacronym{ap}{AP}{access-point}
\newacronym[longplural={user equipment}]{ue}{UE}{user equipment}
\newacronym{hd}{HD}{half-duplex}
\newacronym{ul}{UL}{uplink}
\newacronym{dl}{DL}{downlink}
\newacronym{iid}{i.i.d}{independent and identically distributed}
\newacronym{inai}{InAI}{inter-AP interference}
\newacronym{inui}{InUI}{inter-user interference}
\newacronym{isac}{ISAC}{integrated sensing and communication}
\newacronym{cli}{CLI}{cross-link interference}
\newacronym{los}{LoS}{line-of-sight}
\newacronym{nlos}{NLoS}{non line-of-sight}
\newacronym{mui}{MUI}{multi-user interference}
\newacronym{ula}{ULA}{uniform linear array}
\newacronym{rcs}{RCS}{radar cross-section}
\newacronym{ci}{CI}{clutter-interference}
\newacronym{mrc}{MRC}{maximal ratio combining}
\newacronym{mrt}{MRT}{maximal ratio transmission}
\newacronym{uaf}{UatF}{use-and-then-forget}
\newacronym{glrt}{GLRT}{generalized likelihood ratio test}
\newacronym{ml}{ML}{ maximum likelihood}
\newacronym{pdf}{PDF}{probability density function}
\newacronym{aoa}{AoA}{angle of arrival}
\newacronym{aod}{AoD}{angle of departure}
\newacronym{rzf}{RZF}{regularized zero-forcing}
\newacronym{si}{SI}{self-interference}
\newacronym{bcrlb}{BCRB}{Bayesian CRLB}
\newacronym{crlb}{CRLB}{Cram\'er-Rao lower bound}
\newacronym{bim}{BIM}{Bayesian information matrix}
\newacronym{fim}{FIM}{Fisher information matrix}
\newacronym{wrt}{w.r.t.}{with respect to}
\newacronym{pod}{PoD}{probability of detection}
\newacronym{pfa}{PoFA}{probability of false alarm}
\newacronym{pomd}{PoMD}{probability of miss detection}
\newacronym{cdf}{CDF}{cumulative distribution function}
\newacronym{llr}{LLR}{log-likelihood ratio}
\newacronym{mle}{MLE}{maximum likelihood estimate}
\newacronym{qos}{QoS}{quality-of-service}
\newacronym{tui}{TUI}{treating uplink signal as interference}
\newacronym{ser}{SER}{symbol error rate}
\newacronym{tri}{TrI}{target reflected interference}
\usepackage{url}
\usepackage[hidelinks, breaklinks]{hyperref}
\hypersetup{
	colorlinks   = true, 
	urlcolor     = blue, 
	linkcolor    = blue, 
	citecolor   = blue, 
}
\usepackage[capitalise,nameinlink]{cleveref}  
\usepackage{etoolbox}

\crefformat{equation}{(#2#1#3)} 
\crefname{figure}{Fig.}{Figs.}
\crefname{section}{Sec.}{Secs.}
\usepackage{float}
\usepackage[linesnumbered, ruled, vlined]{algorithm2e}

\SetCommentSty{xCommentSty}
\LinesNumbered
\SetSideCommentRight
\DontPrintSemicolon
\RestyleAlgo{algoruled}

\SetCommentSty{mycommfont}
\SetKwInput{KwInput}{Input}                
\SetKwInput{KwOutput}{Output}            
\SetKwInput{KwInt}{Initialization}

\newenvironment{myfigure*}
{\begin{figure*}[!t]\centering}
	{\hrule\end{figure*}}

\newcommand{\E}[1]{{ \mathbb{E} \big[#1  \big] }}
\newcommand{\Elr}[1]{{ \mathbb{E} \left[#1  \right] }}
\newcommand{\Elrc}[2]{{ \mathbb{E}_{{#2}} \left[#1  \right] }}

\newcommand{\bpr}[1]{{\left( #1 \right)}}

\newcommand{\bsr}[1]{{\left[ #1 \right]}}
\newcommand{\blkd}[1]{\mathtt{Blkd}{\left[ #1 \right]}}
\newcommand{\blkdm}[1]{\mathtt{Blkd}{\left[ #1 \right]_{m=1}^{M_{\mathsf{u}}}}}

\newcommand{\blkdT}[1]{\mathtt{Blkd}{\left[ #1 \right]_{\tau=1}^{T}}}
\newcommand{\bbr}[1]{{\left\{ #1 \right\}}}
\newcommand{\tr}[1]{{\mathtt{Tr} \left\{ #1 \right\}}}
\newcommand{\diag}[1]{{\mathtt{Diag} \left( #1 \right)}}
\newcommand{\norm}[1]{{ \left\Vert #1 \right\Vert_{2} }}
\newcommand{\snorm}[1]{{ \left\Vert #1 \right\Vert_{2}^2 }}

\newcommand{\vect}[1]{{\mathbf{ #1}}}

\newcommand{\abslr}[1]{{\left\lvert {#1}\right\rvert}}
\newcommand{\abs}[1]{{\big\lvert {#1}\big\rvert}}

\newcommand{\cn}[1]{{\mathcal{CN}\left({#1}\right)}}
\newcommand{\mbbC}[1]{{\in\mathbb{C}^{{#1}}}}
\newcommand{\argmax}[1]{{\underset{{#1}}{\mathtt{arg\,max}}}}

\newcommand{\maximize}[1]{{\underset{{#1}}{\mathtt{max}}}}
\newcommand{\minimize}[1]{{\underset{{#1}}{\mathtt{min}}}}
\newcommand{\mt}[1]{\mathtt{#1}}
\newcommand{\p}[1]{p\left({#1}\right)}

\newcommand{\Cov}[1]{\mathtt{Cov}\left\{{#1}\right\}}
\newcommand{\invCov}[1]{\mathtt{Cov}^{-1}\left\{{#1}\right\}}

\newcommand{\Rank}[1]{\mathtt{Rank}\left\{{#1}\right\}}
\newcommand\gldec[2]{
	\underset{#1}{\overset{#2}{\gtrless}}
}
\newcommand{\bigO}[1]{{\mathcal{O}\left( #1 \right)}}

\def\Au{{\mathcal{A}_{\mathtt{u}}}}
\def\Ad{{\mathcal{A}_{\mathtt{d}}}}
\def\Uu{{\mathcal{U}_{\mathtt{u}}}}
\def\Ud{{\mathcal{U}_{\mathtt{d}}}}

\def\A{{\mathcal{A}}}
\def\Hzero{{\mathcal{H}_{0}}}
\def\Hone{{\mathcal{H}_{1}}}
\def\Lm{\mathcal{L}_{m}}
\def\L{\mathcal{L}}

\def\Nvar{N_{0}} 
\def\vf{{\mathbf{f}}}

\def\vp{\mathbf{p}}
\def\vP{\mathbf{P}}

\def\vs{\mathbf{s}}
\def\vh{{\mathbf{h}}}
\def\vH{{\mathbf{H}}}
\def\vx{{\mathbf{x}}}
\def\vy{{\mathbf{y}}}
\def\vI{{\mathbf{I}}}
\def\vZ{{\mathbf{0}}}

\def\vV{{\mathbf{V}}}
\def\ar{{\mathbf{a}_{\mathtt{r}}}}
\def\at{{\mathbf{a}_{\mathtt{t}}}}
\def\mB{{\mathbf{B}}}
\def\mG{{\mathbf{G}}}
\def\mR{{\mathbf{R}}}

\def\bgamma{\boldsymbol{\gamma}}
\def\bSigma{\boldsymbol{\Sigma}}
\def\bPi{\boldsymbol{\Pi}}
\def\bXi{\boldsymbol{\Xi}}
\def\btheta{\boldsymbol{\theta}}
\def\bsymbol{\mathbf{s}_{\mt{u}}}

\def\bA{\mathbf{A}}
\def\bUpsilon{\boldsymbol{\Upsilon}}

\def\bOmega{\boldsymbol{\Omega}}
\def\bOmegam{\boldsymbol{\Omega}_{m}}
\def\mtu{{\mt{u}}}
\def\mtd{{\mt{d}}}
\def\mts{{\mt{s}}}
\def\mtc{{\mt{c}}}
\def\mto{{\mt{o}}}

\def\mtT{{\mt{T}}}
\def\mtD{{\mt{D}}}

\def\mtCPU{{\mt{CPU}}}
\def\SINRuk{{\mt{SINR}_{\mtu, k}}}
\def\SINRdn{{\mt{SINR}_{\mtd, n}}}
\def\SCNRs{{\mt{SCNR}_{\mts}}}
\def\SCNRs{{\mt{SCNR}_{\mts}}}

\def\Rcomm{{\mt{R}_{\mt{com.}}}}
\def\Ed{{\mathtt{p}_{\mt{d}}}}

\def\bEu{\boldsymbol{\mathtt{P}}_{\mt{u}}}
\newcommand{\Eu}[1]{\mathtt{p}_{\mt{u}, {#1}}}
\def\Ku{K_{\mtu}}
\def\Kd{K_{\mtd}}
\def\Md{M_{\mtd}}
\def\Mu{M_{\mtu}}

\def\ttui{\tt{TUI}}
\def\tmle{\tt{MLE}}

\newtheorem{thm}{Theorem}
\newtheorem{lem}[thm]{Lemma}
\newtheorem{prop}[thm]{Proposition}
\newtheorem{cor}[thm]{Corollary}

\newtheorem{rem}{Remark}

\newtheorem{obs}{Observation}
\newcommand{\edit}[1]{{\color{black}{#1}}}

\let\mybibitem\bibitem
\renewcommand{\bibitem}[1]
{\ifstrequal{#1}{Jeong_TVT}{\color{black}\mybibitem{#1}}
{\ifstrequal{#1}{Com_Sense_ISAC}{\color{black}\mybibitem{#1}}
{\ifstrequal{#1}{3gpp2010further}{\color{black}\mybibitem{#1}}
{\ifstrequal{#1}{richards2010principles}{\color{black}\mybibitem{#1}}
{\ifstrequal{#1}{Marzetta_Larsson_Yang_Ngo_2016}{\color{black}\mybibitem{#1}}
{\color{black}\mybibitem{#1}}
}}}}}

\def\BibTeX{{\rm B\kern-.05em{\sc i\kern-.025em b}\kern-.08em
		T\kern-.1667em\lower.7ex\hbox{E}\kern-.125emX}}
\allowdisplaybreaks
\begin{document}
\title{Joint Sensing and Bi-Directional Communication with Dynamic TDD Enabled Cell-Free MIMO}
\author{Anubhab Chowdhury, Sai Subramanyam Thoota, and Erik G. Larsson,~\textit{Fellow,~IEEE}
\thanks{A part of this work has been presented at the International Conference on Communications (ICC), 2025~\cite{Anubhab_Sai_Erik_ICC_2025}.}
\thanks{Anubhab Chowdhury and Erik G. Larsson are with the Dept. of Electrical Engineering (ISY), Linköping University, 58183 Linköping, Sweden. Emails:\{anubhab.chowdhury, erik.g.larsson\}@liu.se}
\thanks{ Sai Subramanyam Thoota was with Linköping University, Dept. of Electrical Engineering (ISY), 58183 Linköping, Sweden, when this work was performed. He is now with Nokia Standards, Bengaluru, India, 560045. Emails: sai.subramanyam.thoota@liu.se / sai.subramanyam\_thoota@nokia.com} 
\thanks{This work was supported in part by the KAW foundation, ELLIIT, and the Swedish Research Council (VR).}
}
\maketitle

\begin{abstract}
This paper studies integrated sensing and communication~(ISAC) with dynamic time division duplex~(DTDD) cell-free~(CF) massive multiple-input multiple-output~(mMIMO) systems. DTDD enables the CF mMIMO system to concurrently serve both uplink~(UL) and downlink~(DL) users with spatially separated \emph{half-duplex~(HD)} access points~(APs) using the same time-frequency resources. Further, to facilitate ISAC, the UL APs are utilized for both UL data and target echo reception, while the DL APs jointly transmit the precoded DL data streams and target signal. In this context, we present centralized and distributed generalized likelihood-ratio tests~(GLRTs) for target detection treating UL users' signals as sensing interference.  We then quantify the optimality and complexity trade-off between distributed and centralized GLRTs and benchmark the respective estimators with the Bayesian Cram\'er-Rao lower bound for target radar-cross section~(RCS). Then, we present a unified framework for joint UL users' data detection and RCS estimation. Next, for communication, we derive the signal-to-noise-plus-interference~(SINR) optimal combiner accounting for the cross-link and radar interference for UL data processing. In DL, we use regularized zero-forcing for the users and propose two types of precoders for the target: one ``user-centric" that nullifies the interference caused by the target signal to the DL users and one ``target-centric" based on the dominant eigenvector of the composite channel between the target and the APs. Finally, numerical studies corroborate with our theoretical findings and reveal that the \emph{GLRT is robust to inter-AP interference, and DTDD doubles the $90\%$-likely sum UL-DL SE compared to traditional TDD-based CF-mMIMO ISAC systems}; while using HD hardware.  
\end{abstract}

\begin{IEEEkeywords}
Bayesian Cram\'er-Rao lower bound, Cell-Free, Dynamic TDD, GLRT, ISAC.
\end{IEEEkeywords}

\section{Introduction}
\IEEEPARstart{P}{hysical-layer} models of the next generation wireless systems are required to support uniformly high \gls{se} for communication users and also provide robust sensing information for several smart applications, thereby offering efficient reuse of limited spectral resources. Thus, \gls{isac} has been incorporated as one of the key features of $6$G physical layer~\cite{ISAC_6G_TWC, Christos_TCOM_2020}, and coupling this with \gls{cf}
\gls{mimo}~\cite{cell_free_small_cells} can offer improvement in communication as well as sensing performance compared to conventional cellular \gls{isac} systems~\cite{Alkhateeb_Asilomar, Demir_Globecom, Huang_TVT}. On a parallel avenue, it has been shown that \gls{dtdd} with \gls{hd} \glspl{ap} offers superior performance compared to \gls{tdd} and even \gls{fd} \glspl{ap} in the presence of asymmetric and heterogeneous
\gls{ul}-\gls{dl} data traffic~\cite{Martin_ICASSP, Martin_Asilomar_2023, Martin_Erik_SPAWC_2024, Anubhab_Chandra_TCOM_22, Anubhab_Chandra_TCOM_2024}. This paper \textit{explores the potential of \gls{isac} for \gls{dtdd} \gls{cf} \gls{mimo}      
systems, addresses several of its signal processing challenges, and underpins the trade-offs between achievable sum \gls{ul}-\gls{dl} \gls{se} and the target detection performance in the presence of \gls{cli}: \gls{inai} and \gls{inui}.\footnote{In \gls{dtdd}, a subset of \glspl{ap} receives \gls{ul} users' data, while the complementary subset transmits to the \gls{dl} users; thus incurring additional interference between the \gls{ul} and \gls{dl} \glspl{ap}~(\gls{inai}) and \gls{ul} and \gls{dl} users~(\gls{inui}). We refer to such interference as \gls{cli}.}}

\subsection{Literature Review} 
Several studies have been conducted in recent years separately on~\gls{isac},~\gls{dtdd}, and \gls{cf}~\gls{mimo}; however, we only pertain to the literature that connects these physical layer methods. Firstly, the key research focus for \gls{isac} has been on waveform design in a single cell mono-static~(i.e., co-located transceivers) environment with the baseline duplexing scheme being \gls{fd}~\cite{Eldar_FD_ISAC, BOttersten_JSTSP_2021, Buzzi_Asilomar_2019, Deepak_Mishra_SPAWC}. In such settings, the authors typically analyze the effects of sensing signal on the achievable \gls{dl} rate and design suitable precoders for joint \gls{dl} communication and sensing. Contrarily, bi-static sensing, with a pair of spatially separated transmit and receive antenna arrays, is more lucrative than mono-static sensing, as it dispenses the need for \gls{fd} capability of the transceivers; whose performance has been analyzed in~\cite{Bi_Static_ISAC_ICC}. This, in turn, obviates the additional signal processing overhead required for self-interference cancellation. Further, multi-static sensing, which consists of multiple non-colocated transmit and receive \glspl{ap}, inherently offers superior macro-diversity and high average \gls{snr} across the geographical area~\cite{Wei_WCL_CRAN_Multi_Static, Seamless_ISAC_TWC}.

Next, considering \gls{cf}-\gls{isac}, the authors in~\cite{Alkhateeb_Asilomar} proposed max-min fair beamforming optimization for \gls{dl} communication and \gls{ul} sensing with distributed \glspl{ap} operating in \gls{tdd}. In~\cite{Demir_Globecom}, the authors proposed a power allocation algorithm to maximize the sensing \gls{snr} with the constraint that a minimal \gls{dl} \gls{sinr} for each user is guaranteed. Parallelly,~\cite{Huang_TVT} aims to minimize the total transmit power while guaranteeing a minimum per-user \gls{sinr} and \gls{crlb} to estimate the target location. Further,\cite{Sofie_Pollin_ICC_2022} provides a bi-static target detection framework for \gls{cf} systems considering a single transmit \gls{ap}. Recently, the authors in~\cite{Ozlem_Lett} present an expository study on distributed target detection considering multi-static sensing. \edit{The work in~\cite{Com_Sense_ISAC} investigates the transmit beamforming optimization for distributed \gls{isac}. Recently,~\cite{Jeong_TVT} considers a perceptive mobile network implementation of \gls{tdd} \gls{cf} \gls{isac}.}

However, we note that all the previous studies considered \edit{\gls{tdd}~\cite{cell_free_small_cells}} as the baseline duplexing scheme for bi/multi-static sensing in \gls{cf}, implying that the \glspl{ap} are utilized either for \gls{dl} data transmission or \gls{ul} sensing; restricting the full resource utilization that can be effectuated if concurrent target detection and \gls{ul} communication can be enabled. In other words, joint \gls{ul} signal estimation and target detection, and the effect of \gls{tri} on \gls{ul} \gls{se}, within the \gls{cf} paradigm has not yet been investigated.

Now, for concurrent \gls{ul}-\gls{dl} communication, one can consider either of the two duplexing schemes: \gls{fd} or \gls{dtdd}. As alluded to earlier, self-interference is a fundamental bottleneck for \gls{fd}; and it incurs additional power-hungry hardware to cancel the self-interference below the noise floor at the receive antenna arrays.
Contrariwise, \gls{dtdd} has been actively considered since LTE Rel.~$12$, therein referred to as enhanced interference mitigation and traffic adaptation, as a candidate duplexing scheme for $5$G and beyond~\cite{Kim_DTDD_Survey}. The main reasons are: (i) \gls{dtdd} obviates the need for self-interference cancellation and attains bi-directional communication with \gls{hd} hardware, making it a low-overhead, feasible alternative to traditional \gls{fd}; (ii) \gls{dtdd} enjoys the benefits of static \gls{tdd} such as reciprocity based beamforming which is critical for scalability; and (iii) offers efficient time-frequency resource management for instantaneous \gls{ul}-\gls{dl} traffic variations~\cite{Mukherjee_Globecom}. 

To this end, recent studies in \gls{dtdd}-enabled \gls{cf} systems show that \gls{hd} \glspl{ap} with appropriate \gls{ul}-\gls{dl} scheduling and power control offer superior performance compared to cellular and even \gls{cf} \edit{\gls{fd} systems}~\cite{Anubhab_Chandra_TCOM_22, Anubhab_Chandra_TCOM_2024, Martin_Asilomar_2023, Hien_JSAC}. 
Here, we note that \gls{dtdd} entails two \glspl{cli}, viz. \gls{inai} and \gls{inui}; nonetheless, these \glspl{cli} can be mitigated via \gls{ap} scheduling~\cite{Zhu_Duplex_Mode_Lett}, power control~\cite{Anubhab_Chandra_TCOM_22, Anubhab_Chandra_TCOM_2024, Hien_JSAC}, and dedicated inter-\gls{ap} and inter-user channel estimation techniques~\cite{Martin_Asilomar_2023, Martin_ICASSP, Martin_Erik_SPAWC_2024}. \edit{However, above studies primarily focused solely on the \gls{ul}-\gls{dl} communication aspect, and thus, integrating sensing capability with \gls{dtdd} \gls{cf} is yet unexplored. Due to the presence of \glspl{cli}, interference from \gls{ul} users' signals for target detection, and \gls{tri} in \gls{ul} \gls{se}; it is not apparent how much benefit can be gleaned using \gls{dtdd} for \gls{isac}. Hence, it is interesting and pertinent to \emph{investigate the mutual trade-off between target detection with/without these \glspl{cli} and quantify the gain in sum \gls{ul}-\gls{dl} \gls{se} because of efficient spectral resource utilization.}}

Motivated by the relevance of \gls{dtdd} and \gls{isac} in \gls{cf} \gls{mimo}, this paper characterizes the dependence of various \gls{cli} on the overall sensing and bi-directional communication performance. To the best of our knowledge, this is the first paper to report the benefits of \gls{dtdd} for \gls{cf} \gls{isac}.
 
\subsection{Contributions} 
We propose a novel framework of \gls{cf} \gls{isac} with \gls{dtdd} for sensing and concurrent \gls{ul}-\gls{dl} communication. We examine: 
$(i)$ target detection in the presence of \gls{ul} users' signals and \gls{inai}; and $(ii)$ bi-directional communication; incorporating the effects of \gls{inai} and \gls{tri} in \gls{ul} \gls{se}; and \gls{inui} in \gls{dl} \gls{se}. 
Our key contributions are\footnote{\edit{In contrast to~\cite{Anubhab_Sai_Erik_ICC_2025}, the results herein cover a  general framework for, including several additional developments, such as \gls{bcrlb}, radar \gls{scnr}, and distributed \gls{glrt};   \cite{Anubhab_Sai_Erik_ICC_2025} can be subsumed as a special case.}}:
\begin{enumerate}
    \item We first derive centralized and distributed estimators for target \gls{rcs} and develop respective \glspl{glrt} treating \gls{ul} users' signals as part of the equivalent sensing noise; which also includes \gls{inai} and \gls{awgn}. 
    We call this scheme \gls{tui}~(see~\cref{sec: TUI}).  
    \begin{itemize}    
    \item We analytically quantify the loss in detection performance due to distributed processing and highlight the conditions under which the distributed scheme attains centralized performance~(see~\cref{thm: centralize_vs_distributed}).  
    \item We then derive~\gls{bcrlb} to benchmark the performance of proposed centralized and distributed \gls{rcs} estimators~(see~\cref{lemm: BCRLB} and~\cref{lem: dist_est}).
    \end{itemize}
    \item We next develop a framework for joint \gls{ul} data and~\gls{rcs} estimation, validate its robustness to \gls{inai}, and compare it with the baseline \gls{tdd}~(see~\cref{sec: Joint} and~\cref{fig: ser_set1}).
    \item Next, we present sum \gls{ul}-\gls{dl} \gls{se}, considering \gls{sinr} optimal combiner in \gls{ul} and \gls{rzf} precoder for \gls{dl} users. For the target, we propose two 
     precoders: one ``{user-centric}" that nullifies the interference caused by the target signal to the \gls{dl} users and one ``{target-centric}" based on the dominant eigenvector of the composite channel between the target and the \glspl{ap}. The proposed precoders are agnostic of the underlying channel models between the target and the \gls{ul}/\gls{dl} \glspl{ap}; thus, generalized compared to the model-dependent precoding presented in~\cite{Hien_ISAC_Globcom_2023, Ozlem_Lett}~(see~\cref{sec: SE and SINR}). 
 \item We then derive average \edit{sensing \gls{scnr}} that captures the effects of target \gls{rcs}, choice of \gls{dl} precoder, \gls{ap} scheduling, and \gls{inai} on the sensing \gls{se}~(see~\cref{lemma: sensing_SINR}). 
    \item Finally, we develop a low complexity \gls{ap} scheduling algorithm based on the local \gls{ul}/\gls{dl} traffic load and target-location, which solves an exponentially complex \gls{ap} mode selection problem in polynomial time~(see~\cref{sec: Scheduling}). 
\end{enumerate}

Finally, we present extensive numerical results to validate our theoretical findings. 
The key takeaway is that \gls{dtdd}-enabled \gls{cf} is a promising solution for \gls{isac} for \emph{two} fundamental reasons: (a) target detection performance is similar to the case when the received signal is not corrupted by the \gls{ul} signals, because of the robustness of the \gls{glrt} to residual \gls{inai}~(see~\cref{fig: TDD_DTDD_RoC}); (b) communication performance is far more superior~(see~\cref{fig: SE1}) compared to the baseline \gls{tdd} based multi-static sensing system~\cite{Demir_Globecom, Hien_ISAC_Globcom_2023, Bi_Static_ISAC_ICC}  due to bidirectional transmission/reception capability of the overall system. 

\subsubsection*{Notation} 
 For $\vect{A} \mbbC{M\times N}$ and $\vect{a} \mbbC{N}$; $\bsr{\vect{A}}_{:,n}$, $\bsr{\vect{A}}_{m,n}$, and $\bsr{\vect{a}}_{n}$ denote their $n$th column, $\bpr{m,n}$th entry, and $n$th element, respectively.
 $\diag{x_{1},x_{2},\ldots, x_{N}}$ denotes a diagonal matrix with $n$th diagonal entry being $x_{n}$. $\blkdm{\vect{A}_m}\mbbC{N\Mu\times N\Mu}$ is a block-diagonal matrix with $m$th block being $ \vect{A}_{m}\mbbC{N\times N}, m=1,2, \ldots, \Mu$. $\Rank{.}$, $\tr{.}$, and $``\otimes"$ indicate rank, trace, and Kronecker product, respectively.  Operators $'\cup'$, $'\cap'$, $'\backslash'$, and $'\abs{.}'$ over sets~(denoted by calligraphic letters) indicate union, intersection, set difference, and cardinality, respectively. \edit{$\E{\vx}$ denotes expectation \gls{wrt} $\vx$. 
 $\Elrc{\cdot \vert\vx}{\vy}$ denotes expectation \gls{wrt}  $\vy$, conditioned on $\vx$.}
$\vx\sim\cn{\vZ_N,\Cov{\vx}}$ indicates that 
 $\vx$ is \gls{cn} with zero mean and covariance matrix $\Cov{\vx}\mbbC{N\times N}$. Other recurrent notations used in the paper are summarized in~\cref{tables: Notation}.

 \begin{table}[!]
 	\bgroup
 	\def\arraystretch{1.35} \setlength{\arrayrulewidth}{0.01pt}
 \begin{center}
 \begin{tabular}{| c | c |}
 \hline
 \rowcolor[HTML]{9AFF99} 
 Notation  & Definition\\
 \hline
 $\vh_{mk}$ & Channel between $m$th AP and $k$th user\\ 
 \hline
 $\vx_{\mtd, j}$ & \Gls{dl} precoded data and sensing signal at $j$th \gls{ap}\\
 \hline
 $\mtT_{m}^{ \ttui}$ & \Gls{llr} at the $m$th AP\\
 \hline 
 $\mtT^{ \ttui}$ & Global \gls{llr} at the \gls{cpu}\\
 \hline
 $\mtT_{\mtCPU}^{ \ttui}$ &  Fused \Gls{llr} at the \gls{cpu}~($\sum\nolimits_{m\in\Au}\mtT_{m}^{ \ttui}$)\\
 \hline
  $\hat{\bgamma}_{m}^{ \ttui}$ &  Estimate of \gls{rcs} at $m$th \gls{ap} \\
 \hline
 $\hat{\bgamma}^{ \ttui}$ & Estimate of \gls{rcs} at the \gls{cpu} \\ 
 \hline
 $\dot{\mR}_{mj}$ & Channel between $j$th \gls{dl} \gls{ap}-target-$m$th \gls{ul} \gls{ap} \\
 \hline
 $\gamma_{mj}$ & \gls{rcs} of target ($j$th \gls{dl} \gls{ap}-target-$m$th \gls{ul} \gls{ap})\\
 \hline
 $\ddot{{\mR}}_{m}[\tau]$ & Measurement matrix at $m$th \gls{ap} for target detection \\
 \hline
 $\ddot{{\mR}}[\tau]$ & Measurement matrix at the \gls{cpu} for target detection \\
 \hline
 $\tilde{\mG}_{mj}$ & Residual \gls{inai} channel $j$th \gls{dl} \gls{ap}-$m$th \gls{ul} \gls{ap}\\
 \hline
 $\bSigma_{\mts}[\tau]$ & Equivalent sensing noise covariance at the \gls{cpu}\\
 \hline
 $\bSigma_{\mts, m}[\tau]$ & Equivalent sensing noise covariance at the $m$th \gls{ap}\\
 \hline
 $\widetilde{\btheta}_{m}$ & Concatenated \gls{ul} signals and target \gls{rcs} at $m$th \gls{ap}\\
 \hline
 $ \btheta_{\tt CPU}$ & Concatenated \gls{ul} signals and target \gls{rcs} at the \gls{cpu}\\
 \hline
 \end{tabular}
 \caption{Definitions of relevant symbols used throughout the paper.}
 \label{tables: Notation}
 \end{center}
\egroup
 \end{table}

\section{System Model}\label{sec: system_model}
 We consider a \gls{dtdd} \gls{cf} system where \gls{hd} \gls{ul} and \gls{dl} \glspl{ap}, each equipped with $N$ antennas, jointly serve single antenna \gls{ul} and \gls{dl} communication users. Further, the \gls{dl} \glspl{ap} transmit additional sensing signals whose reflected echoes are used by the \gls{ul} \glspl{ap} to detect the presence of a target. 
 Further, the \glspl{ap} are fully synchronized\footnote{Although phase synchronization of distributed \glspl{ap} is a critical issue for multi-static sensing, recent advancements for phase calibration of distributed antennas made \gls{cf} practically viable technology~\cite{Erik_TSP_Massive_Synch}.} and connected to the \gls{cpu} via error-free front-haul links. Next, let $\Au$ and $\Ad$ be the index sets of \glspl{ap} scheduled in \gls{ul} and \gls{dl}, respectively, with $\Au\cup\Ad=\A$, where $\A$ is the set of all \gls{hd} \glspl{ap}; $\abs{\A}=M$. Also, let $\abs{\Au}=M_{\mtu}$, and $\abs{\Ad}=M_{\mtd}$.\footnote{The analysis presented for the target detection assumes a scheduled set of \glspl{ap} for easy understanding. How these \glspl{ap} are selected will be discussed later in the manuscript.} 
 Similarly, $\Uu$ and $\Ud$ are the index sets of \gls{ul} and \gls{dl} users, respectively; with $\abs{\Uu\cup\Ud}=K$, $\abs{\Uu}=K_{\mtu}$, and $\abs{\Ud}=K_{\mtd}$.  The users' locations and corresponding modes/data demands~(i.e., \gls{ul}/\gls{dl}) can be made available at the \gls{cpu} via initial channel sounding reference signals. For an illustration of the overall system see~\cref{fig: system_model}, where we also indicate the additional \glspl{cli}. 
	\begin{figure}[t!]
		\centering
\includegraphics[width=0.67\linewidth]{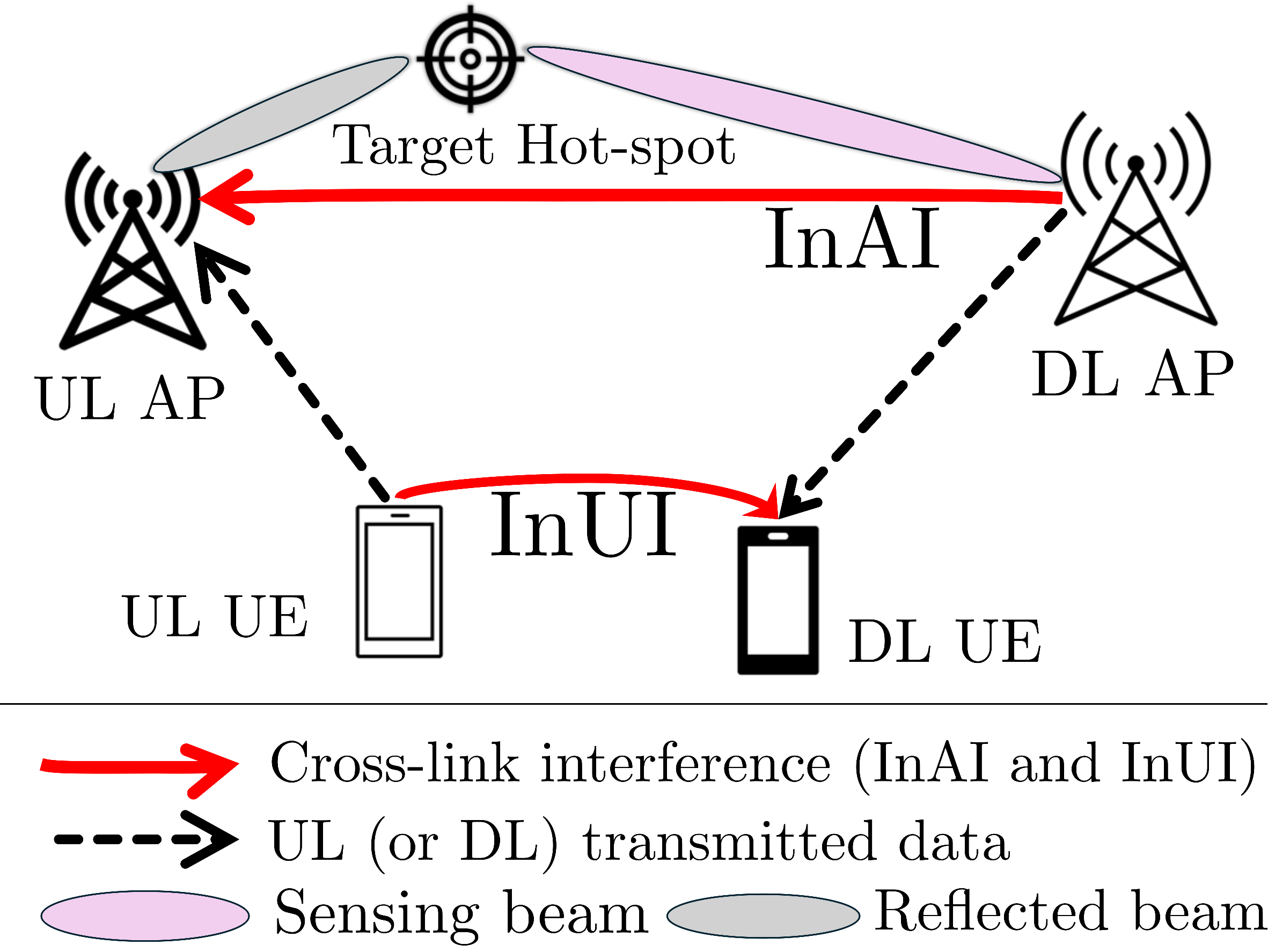}
		\caption{\edit{System model for \gls{dtdd} \gls{cf} \gls{mimo} with \gls{isac} capability. The \gls{cli} 
 (viz. \gls{inai} and \gls{inui}) are denoted by red arrows.}}\label{fig: system_model}
	\end{figure}

\subsubsection{Communication Channel Model}\label{sec:channel_models}	

The \gls{ul} channel from the $k$th user to the $m$th \gls{ap} is denoted by $\vh_{mk}\in\mathbb{C}^N$, and \gls{dl} channels are assumed to be reciprocal of the \gls{ul} channels.
Further, $\vh_{mk}=\sqrt{\beta_{mk}}\vf_{mk} \in \mathbb{C}^N$, where $\beta_{mk}>0$ captures the effect of large-scale fading~(i.e., path loss and shadowing) which remains unchanged over several channel coherence intervals and is known to the \glspl{ap} and the \gls{cpu}~\cite{cell_free_small_cells, making_cell_free_TWC}; and $\vf_{mk}$ is the block-fading channel. \footnote{\edit{The subsequent analysis depends on the abstract channel vector $\vf_{mk}$, but not on the specific modeling of it.}}
Also, let $\vh_{k}=\bsr{\vh_{1k}^{T}, \vh_{2k}^{T}, \ldots, \vh_{M_{\mtu}k}^{T}}^T\mbbC{M_{\mtu}N}$ be the concatenated channel vector from the $k$th user to \gls{ul} \glspl{ap} at the \gls{cpu}.

Next, we note that the inter-\gls{ap} channels remain constant for several coherence intervals, whose \gls{csi} can be made available to the \gls{cpu} before the data transmission phase. However, the inter-\gls{ap} \gls{csi} can be erroneous, and we denote the residual \gls{inai} channel from the $j$th \gls{dl} \gls{ap} to the $m$th \gls{ul} \gls{ap} by $\tilde{\mG}_{mj}$.   Finally, let $\mathtt{g}_{nk}$ denote the channel between $k$th \gls{ul} user and the $n$th \gls{dl} user.

\subsubsection{Sensing Channel Model}\label{subsection: sensing_channel_model}	
 Let $\mR_{mj}\triangleq \gamma_{mj}\dot{\mR}_{mj}$ denote the composite channel from the $j$th \gls{dl} \gls{ap} to the $m$th \gls{ul} \gls{ap} via the sensing zone. Here, $\gamma_{mj}\sim\cn{0,\sigma_{mj}}$ denotes the \gls{rcs} associated with the target.\footnote{We note that depending on the physical properties of the target, \gls{rcs} can also be deterministic. However, analysis with a prior generalizes to the case when $\gamma_{mj}$ is a deterministic unknown. Hence, the analysis in~\cite{Anubhab_Sai_Erik_ICC_2025} can be derived as a special case of the results presented in the sequel.}
 Further, $\dot{\mR}_{mj}$ encases the effect of round-trip path loss between the transmit and receive \glspl{ap} via the target and known at the \glspl{ap} and the \gls{cpu}.\footnote{The \gls{dl} \glspl{ap} transmit the intended target signals to a predetermined hot-spot area during a given observation window, and thus, the associated \glspl{aoa} and \glspl{aod} along with the path loss coefficients~(i.e., constituents of $\dot{\mR}_{mj}$) can be pre-calibrated~\cite{Ozlem_Lett}.} \edit{Specifically, $\dot{\mR}_{mj}$ is characterized by the \glspl{aod}~($\phi_{j}$) and \glspl{aoa}~($\phi_m$) as
$\dot{\mR}_{mj}=\sqrt{\beta_{\mts, mj}}\ar(\phi_{m})\at(\phi_{j})^{T},$
where $\at(\phi_{j}) = \bsr{1, e^\bbr{\iota\pi\sin(\phi_{j})}, \ldots, e^\bbr{(N-1)\iota\pi\sin(\phi_{j})}}^{T}$ and $\ar(\phi_{m}) =\bsr{1, e^\bbr{\iota\pi\sin(\phi_{m})}, \ldots, e^\bbr{(N-1)\iota\pi\sin(\phi_{m})}}^{T}$.
Next, the composite two-way path loss $\beta_{\mts, mj}$ is modeled as~\cite{Demir_Globecom, richards2010principles} $\beta_{\mts, mj}={\lambda_{\mathtt c}^2}/{\bpr{4\pi}^3d_{m}^2 d_{j}^2},$ where $d_j$ is the distance between the transmitting AP $j$ and the target location, $d_m$ is the distance between the receiving $m$th AP and the target location, and $\lambda_{\mathtt c}$ is the carrier wavelength. Finally, $\iota=\sqrt{-1}$.}
 
Next, we describe the signal model for \gls{ul}-\gls{dl} communications and the sensing, based on which~\gls{glrt} and corresponding~\glspl{sinr} will be derived.

\subsection{UL and DL Signal Model}
	We start with the \gls{dl} signaling model for convenience. The transmitted signal from the $j$th \gls{dl} \gls{ap} can be expressed as:
 \begin{align}
	&\vx_{\mtd, j}=\sqrt{\Ed}\vP_{j}{\bPi}_{j}^{\frac{1}{2}}\vs_{\mtd}\notag\\&=\sqrt{\Ed}\left\{\sum\nolimits_{n\in\Ud}\sqrt{\pi_{\mtd, jn}}\vp_{\mtd, jn}s_{\mtd, n}+\sqrt{\pi_{\mts, j}}\vp_{\mts, j}s_{\mts}\right\},\label{eq: dl_tx}
\end{align}
 where $\Ed$ is the total radiated power in the \gls{dl} and $\vP_{j}\triangleq\bsr{\vp_{\mtd, j1}, \ldots, \vp_{\mtd, jK_{\mtd}}, \vp_{\mts, j}}$ is the \gls{dl} precoding vectors for users and the target. The coefficients $\pi_{\mtd, jn}$ and $\pi_{\mts, j}$ denote the fraction of $\Ed$ dedicated to the $n$th \gls{dl} user and the sensing zone, respectively, with ${\bPi}_{j}\triangleq\diag{{{\pi_{\mtd, j1}},{\pi_{\mtd, j2}},\ldots, {\pi_{\mtd, jK_{\mtd}}}, {\pi_{\mts, j}}}}$.
Mutually independent scalars $s_{\mtd, n}\sim\cn{0,1}$ and $s_{\mts}\sim\cn{0,1}$ denote the intended signals for the $n$th \gls{dl} user and the sensing zone, respectively, with $\vs_{\mtd}=\bsr{s_{\mtd, 1}, s_{\mtd, 2}, \ldots, s_{\mtd, K_{\mtd}}, s_{\mts}}^{T}$. Thus, the symbols satisfy $\E{\abs{s_{\mtd, n}}^2}=\E{\abs{s_{\mts}}^2}=1$, $\E{s_{\mtd, n}s_{\mts}^*}=0$, and $\E{s_{\mtd, n}s_{\mtd, n'}^*}=0, \forall n\neq n'; n, n'\in\Ud$. The transmitted composite \gls{dl} signal power satisfies $\Elr{\snorm{\vx_{\mtd, j}}}\leq \Ed$. Now, the received signal at the $n$th \gls{dl} user can be expressed as:
	\begin{multline}
		{r}_{\mtd, n}
=\sum\nolimits_{j\in\Ad}\sqrt{\Ed\pi_{\mtd, jn}}\vh_{jn}^{T}\vp_{\mtd, jn}s_{\mtd, n}\\+\sum\nolimits_{j\in\Ad}\sum\nolimits_{n'\in\Ud\backslash\{n\}}\sqrt{\Ed\pi_{\mtd, jn'}}\vh_{jn}^{T}\vp_{\mtd,jn'}s_{\mtd, n'}\\+\sum\limits_{j\in\Ad}\sqrt{\Ed\pi_{\mts, j}}\vh_{ jn}^{T}\vp_{\mts, j}s_{\mts}+\sum\limits_{k\in\Uu}\sqrt{\Eu{k}}\mathtt{g}_{nk}s_{\mtu, k}+w_{\mtd, n}.\label{eq: receive_signal_ue}
	\end{multline}
	The second and third terms in~\eqref{eq: receive_signal_ue} correspond to \gls{dl}-\gls{mui} and the interference due to the sensing signal, respectively. The fourth term corresponds to \gls{inui}, which is an artifact of \gls{dtdd}, where $s_{\mtu, k}\sim\cn{0,1}$ is \gls{ul} signal transmitted by the $k$th user, with $\E{\abs{s_{\mtu, k}}^2}=1$ and $\E{{s_{\mtu, k}s_{\mtu, k'}^*}}=0$ for all $k\neq k'$ and $k, k'\in\Uu$. Here, ${\Eu{k}}$ denotes the \gls{ul} \edit{transmit} power of the $k$th user. Finally, $w_{\mtd, n}\sim\cn{0,\Nvar}$ is the receiver \gls{awgn} at the $n$th \gls{dl} user. 
	
	Next, the signal received at the $m$th \gls{ul} \gls{ap} is
	\begin{align}
		&\edit{\mathbf{r}_{\mtu, m}}=\sum\nolimits_{k\in\Uu}\sqrt{\Eu{k}}\vh_{mk}s_{\mtu, k}\notag\\&+\sum\nolimits_{j\in\Ad}\gamma_{mj}\dot{\mR}_{mj}\vx_{\mtd, j}+\sum\nolimits_{j\in\Ad}\tilde{\mG}_{mj}\vx_{\mtd, j}+\vect{w}_{\mtu, m},\label{eq: uplink_signal_m}
	\end{align}
	where $\vect{w}_{\mtu, m}\sim\cn{\vZ_{N}, \Nvar\vI_{N}}$ is the receiver \gls{awgn} at the $m$th \gls{ul} \gls{ap}. In~\eqref{eq: uplink_signal_m}, the first term corresponds to the communication signals from the users, and the second term is the reflected echo from the target. The third term is \gls{inai}, which contains both sensing and \gls{dl} communication signals. \edit{The effects of scatterers will later be discussed in~\cref{{sec: scatter_model}}.}
 
At this point, we emphasize that \textit{the key difference between the proposed framework and existing \gls{tdd}-based \gls{isac} systems~\cite{Alkhateeb_Asilomar, Demir_Globecom, Christos_TCOM_2020} is that we utilize the same set of \glspl{ap} for both \gls{ul} data processing and target detection; while the later considered exclusive \glspl{ap} for target detection. Further, \gls{dtdd} enables bidirectional communication and sensing with \gls{hd} \glspl{ap}; while in \gls{tdd} \gls{ul} communication and \gls{dl} communication-plus-sensing occur in orthogonal time slots.}

\section{Centralized and Distributed GLRTs with TUI}\label{sec: TUI}
We observe that the received signal in~\eqref{eq: uplink_signal_m} involves two unknowns: \gls{ul} users' data and the \gls{rcs} associated with the target. We first present a target detection scheme that treats the transmitted \gls{ul} users' signals as a part of the equivalent sensing noise; which also includes \gls{inai} and the \gls{awgn}. 

Let $T$ be the length~(in number of transmitted symbols per channel use) of the observation window for target detection. \edit{The target detection scheme can sweep over a large area by beamforming sensing signals to different hot-spots in each interval of length $T$.} For a given slot, indexed by $\tau$, the received signal at the $m$th \gls{ul} \gls{ap} can be expressed as~(see \eqref{eq: uplink_signal_m} for reference):
\begin{align}
		&\mathbf{r}_{\mtu, m}[\tau]=\ddot{{\mR}}_{m}[\tau]\bgamma_{m}+\vect{w}_{\mtu, m}^{\mts}[\tau],\label{eq: uplink_signal_mthAP}
	\end{align}
	where 
 \begin{align}
 \ddot{{\mR}}_{m}[\tau]\triangleq\bsr{\dot{\mR}_{m1}\vx_{\mtd, 1}[\tau],\ldots, \dot{\mR}_{m M_{\mtd}}\vx_{\mtd,  M_{\mtd}}[\tau]}.\label{eq: R_ddt_m}
 \end{align}
 Further, $\bgamma_{m}\triangleq\bsr{\gamma_{m1}, \ldots, \gamma_{m M_{\mtd}}}^{T}$, and the equivalent sensing noise $\vect{w}_{\mtu, m}^{\mts}[\tau]$ is
	\begin{align*}
		\vect{w}_{\mtu, m}^{\mts}[\tau]=\hspace*{-1.7mm}\sum\limits_{k\in\Uu}\sqrt{\Eu{k}}\vh_{mk}s_{\mtu, k}[\tau]+\hspace*{-1.7mm}\sum\limits_{j\in\Ad}\tilde{\mG}_{mj}\vx_{\mtd, j}[\tau]+\vect{w}_{\mtu, m}[\tau].
	\end{align*}
 
Given this framework, we develop two schemes: $(i)$ a distributed scheme wherein \gls{ul} \glspl{ap} evaluate their \glspl{llr} locally and relay them to the \gls{cpu} to obtain a global \gls{llr}; which incurs minimal front-haul overhead. $(ii)$ a centralized scheme, where the \gls{ul} \glspl{ap} relay the $N$-dimensional received vector; and then \gls{cpu} finds the global \gls{llr} based on $NT\Mu$ dimensional received vector. 
 
\subsection{Distributed Scheme}
 Let $\widetilde{\mathbf{r}}_{\mtu, m}\mbbC{NT}$ and $\widetilde{\vect{w}}_{\mtu, m}^{\mts}\mbbC{NT}$ be the concatenated received signal and the effective noise at the end of the observation window. Then, we have the following two hypotheses:
	\begin{subequations}
		\begin{align}
			&\Hzero: \widetilde{\mathbf{r}}_{\mtu, m}=\widetilde{\vect{w}}_{\mtu, m}^{\mts}, \\&
			\Hone: \widetilde{\mathbf{r}}_{\mtu, m}=\widetilde{\bgamma}_{m}+\widetilde{\vect{w}}_{\mtu, m}^{\mts},
		\end{align}
	\end{subequations}
	where 
 \begin{multline}
\widetilde{\bgamma}_{m}\triangleq\bsr{\bpr{\ddot{{\mR}}_{m}[1]\bgamma_{m}}^{T}, \ldots, \bpr{\ddot{{\mR}}_{m}[T]\bgamma_{m}}^{T}}^{T}\\=\bsr{\bpr{\ddot{{\mR}}_{m}[1]}^{T}, \ldots, \bpr{\ddot{{\mR}}_{m}[T]}^{T}}^{T}\bgamma_{m}. 
	\end{multline}
 Here, the null hypothesis $\Hzero$ represents the case where there is no target in the sensing area, and the alternative hypothesis $\Hone$ indicates the presence of a target in the sensing area. Hence, the \gls{rcs} estimation and hypothesis testing problem becomes
	\begin{align}
\bbr{\hat{\bgamma}_{m},\hat{\mathcal{H}}}=\argmax{\bbr{\bgamma_{m},\mathcal{H}}} \quad\p{\bgamma_{m},\mathcal{H}\big\vert  \widetilde{\mathbf{r}}_{\mtu, m}}.
	\end{align}
	Correspondingly the \gls{llr} at the $m$th \gls{ap} can be written as $\ln\Lm$, with 
	\begin{align}
\Lm=\dfrac{\maximize{\bgamma_{m}}~\p{ \widetilde{\mathbf{r}}_{\mtu, m}\vert \bgamma_{m}, \Hone}\p{\bgamma_{m}\vert \Hone}}{\p{ \widetilde{\mathbf{r}}_{\mtu, m}\vert \Hzero}}.\label{eq: GLRT}
	\end{align}
To derive $\Lm$ in closed-form, we need to evaluate the covariance of $\vect{w}_{\mtu, m}^{\mts}[\tau]$, denoted by $\bSigma_{{\mts, m}}[\tau]$.\footnote{\edit{To be specific, $\bSigma_{{\mts, m}}[\tau]$ is a conditional covariance, conditioned on precoded data~($\vx_{\mtd, j}[\tau]$) and availability of \gls{ul} user to \gls{ap} channels.}} We can show that 
\begin{align}
\bSigma_{{\mts, m}}[\tau]=\sum\nolimits_{k\in\Uu}{\Eu{k}}\vh_{mk}\vh_{mk}^{H}+\bA_{m}[\tau],\label{eq: Sigma_m_APs}
\end{align} 
with $\bA_{m}[\tau]$ being $\bpr{\sum\nolimits_{j\in\Ad}\zeta_{mj}\snorm{\vx_{\mtd, j}[\tau]}+\Nvar}\vI_{N}$, where $\zeta_{mj}$ captures the effects of both the large-scale fading and the power of the residual \gls{inai}~\cite{Bai_Sabharwal_TWC_2017}. For proof, refer to~\cref{app: noise covariance}.
We now have the following proposition. 
 \begin{prop}\label{prop: GLRT1}
 The estimate of the target \gls{rcs} at the $m$th \gls{ap}, denoted by $\hat{\bgamma}_{m}^{\ttui}$, can be evaluated as:
	\begin{align}
\hat{\bgamma}_{m}^{ \ttui}=\bUpsilon_{m}^{ \ttui}\bpr{\sum\nolimits_{\tau=1}^{T}\ddot{{\mR}}_{m}^{H}[\tau]\bSigma_{{\mts, m}}^{-1}[\tau]\mathbf{r}_{\mtu, m}[\tau]},\label{eq: gamma_estimate_MMSE} 
	\end{align}
 with $\bUpsilon_{m}^{ \ttui}\triangleq\left(\sum\nolimits_{\tau=1}^{T}\ddot{{\mR}}_{m}^{H}[\tau]\bSigma_{{\mts, m}}^{-1}[\tau]\ddot{{\mR}}_{m}[\tau]+\bSigma_{\bgamma_{m}}^{-1}\right)^{-1}$, $\bSigma_{\bgamma_{m}}\triangleq\Cov{\bgamma_{m}}=\diag{\sigma_{m1}, \sigma_{m2}, \ldots, \sigma_{m\Md}}$. Hence, the test metric at the $m$th \gls{ap}, denoted by  $\mtT_{m}^{ \ttui}$, is
		\begin{align}
		\mtT_{m}^{ \ttui}=\bpr{\sum\nolimits_{\tau=1}^{T}\ddot{{\mR}}_{m}^{H}[\tau]\bSigma_{{\mts, m}}^{-1}[\tau]\mathbf{r}_{\mtu, m}[\tau]}^{H}\hat{\bgamma}_{m}^{ \ttui}.\label{eq: T_m_GLRT_prior}
		\end{align}
  The \gls{cpu} combines the local \glspl{llr} to find a global test for target detection as follows:
	\begin{align}
\mtT_{\mtCPU}^{ \ttui}=\sum\nolimits_{m\in\Au}\mtT_{m}^{ \ttui}\gldec{\Hzero}{\Hone}\ln{\lambda_{m}},\label{eq: GLRT_dist}
	\end{align}
 where $\lambda_{m}$ is the detection threshold for a given \gls{pfa}.
 \end{prop}
 \begin{proof}
     Proof is omitted for brevity.
 \end{proof}
\begin{obs}
   We note that~\cref{prop: GLRT1} can be simplified to the case when $\bgamma$ is a deterministic unknown. In that case, we can show that \gls{mle} of $\bgamma_{m}$ can be computed as $
\hat{\bgamma}_{m}^{\tmle}=\left(\sum\limits_{\tau=1}^{T}\ddot{{\mR}}_{m}^{H}[\tau]\bSigma_{{\mts, m}}^{-1}[\tau]\ddot{{\mR}}_{m}[\tau]\right)^{-1}\bpr{\sum\limits_{\tau=1}^{T}\ddot{{\mR}}_{m}^{H}[\tau]\bSigma_{{\mts, m}}^{-1}[\tau]\mathbf{r}_{\mtu, m}[\tau]},$
substituting which on~\eqref{eq: T_m_GLRT_prior}, we get the \gls{llr}. However, we note that $\Rank{\ddot{{\mR}}_{m}^{H}[\tau]\bSigma_{{\mts, m}}^{-1}[\tau]\ddot{{\mR}}_{m}[\tau]}$ is upper-bounded by $\min\bbr{\Rank{\ddot{{\mR}}[\tau]},\Rank{\bSigma_{{\mts, m}}^{-1}[\tau]}}$. Now, $\bSigma_{{\mts, m}}[\tau]$ is full rank, with the rank being $N$. However, it is easy to see that $\Rank{\ddot{{\mR}}[\tau]}\leq \min\bbr{N, \Md}$. Hence, if $N<\Md$, which is usually the case in \gls{cf} systems, we can argue that
    $\Rank{\ddot{{\mR}}_{m}^{H}[\tau]\bSigma_{{\mts, m}}^{-1}[\tau]\ddot{{\mR}}_{m}[\tau]}\leq N.$ Thus, $T$ should be at least $\left\lceil{\dfrac{\Md}{\Rank{\ddot{{\mR}}_{m}^{H}[\tau]\bSigma_{{\mts, m}}^{-1}[\tau]\ddot{{\mR}}_{m}[\tau]}}}\right\rceil$ for invertibility of the overall matrix for the computation of $\hat{\bgamma}_{m}^{\tmle}$, although it is not guaranteed that inverse would exist. On the contrary, such issues are avoided for $\hat{\bgamma}_{m}^{\ttui}$ because of diagonal loading by the prior covariance matrix. 
\end{obs}
 
\begin{rem}
    We note that~\cref{prop: GLRT1} assumes perfect \gls{csi} for $\vh_{mk}$ while evaluating the sensing covariance $\bA_{m}[\tau]$. However, we highlight that this is not a limitation of any analysis or result presented in the paper. One can extend ~\cref{prop: GLRT1} and all subsequent results for statistical or trained \gls{csi} for \gls{ap}-user channels by marginalization of conditional covariances, conditioned on \gls{ul} users' signals. Essentially, for statistical \gls{csi}, $\bSigma_{{\mts, m}}[\tau]$ can be evaluated as:
    \begin{align*}
        \bSigma_{{\mts, m}}[\tau]=\Elr{\Cov{\vect{w}_{\mtu, m}^{\mts}[\tau]\vert s_{\mtu, k}[\tau], \forall k}},
    \end{align*}    
    where the expectation is evaluated over the distribution of the \gls{ul} users' transmitted signals. However, this entails further notational book-keeping, adding little to the theme of the paper, which is to understand the impact of \glspl{cli} and imperfect \gls{csi} of such \gls{cli} channels on \gls{isac}, and the trade-off of bi-directional communication with \gls{cli} and \gls{tdd} based \gls{isac} systems without  \gls{inai}, \gls{inui}, and \gls{tri}. 
\end{rem}

\subsection{Centralized Scheme}
Here, \gls{cpu} accumulates $T$ snapshots of the \gls{ul} received signals from all the \gls{ul} \glspl{ap}; which can be expressed as:
 \begin{align}	
 \mathbf{r}_{\mtu}[\tau]=\ddot{{\mR}}[\tau]\bgamma+\vect{w}_{\mtu}^{\mts}[\tau]\mbbC{N M_{\mtu}},\label{eq: central_CPU_UL_snapshot}
	\end{align}
where $\vect{w}_{\mtu}^{\mts}[\tau]$ indicates the equivalent sensing noise, defined as
$\vect{w}_{\mtu}^{\mts}[\tau]=\sum\nolimits_{k\in\Uu}\sqrt{\Eu{k}}\vh_{k}s_{\mtu, k}[\tau]+\sum\nolimits_{j\in\Ad}\tilde{\mG}_{j}\vx_{\mtd, j}[\tau]+\vect{w}_{\mtu}[\tau]$. 
The overall sensing channel matrix at the \gls{cpu} can be written as $\ddot{{\mR}}[\tau]=\blkdm{\ddot{{\mR}}_{m}[\tau]}\mbbC{N M_{\mtu}\times M_{\mtd} M_{\mtu}},$
where
$\ddot{{\mR}}_{m}[\tau]$ is defined in~\eqref{eq: R_ddt_m}.
Also, let $\bgamma=\bsr{\bgamma_{1}^{T},\bgamma_{2}^{T}, \ldots, \bgamma_{M_{\mtu}}^{T}}^{T}\mbbC{M_{\mtd} M_{\mtu}}$, $\tilde{\mG}_{j}\triangleq\bsr{\tilde{\mG}_{1j}^{T}, \ldots, \tilde{\mG}_{M_{\mtu} j}^{T}}^{T}$, and $\vect{w}_{\mtu}[\tau]=\bsr{\vect{w}_{\mtu, 1}^{T}[\tau], \ldots, \vect{w}_{\mtu, M_{\mtu}}^{T}[\tau]}^{T}\mbbC{NM_{\mtu}}$.  The sensing noise at \gls{cpu} is distributed as $\cn{\vZ_{\Mu N}, \bSigma_{\mts}[\tau]}$, with \edit{$\bSigma_{\mts}[\tau]\bpr{\triangleq\Elrc{\vect{w}_{\mtu}^{\mts}[\tau]\vect{w}_{\mtu}^{\mts H}[\tau] \big\vert \bbr{\vh_{k}, \vx_{\mtd, j}[\tau]}}{\vect{w}_{\mtu}^{\mts}[\tau]}}$, being} 
\begin{align}
   \bSigma_{\mts}[\tau]=\sum\nolimits_{k\in\Uu}{\Eu{k}}\vh_{k}\vh_{k}^{H}+\bA[\tau],\label{eq: sensing_noise_cov} 
\end{align}
with 
$\bA[\tau]\triangleq\sum\nolimits_{j\in\Ad}\left(\diag{\zeta_{1j},\ldots,\zeta_{M_{\mtu}j}}\otimes\vI_{N}\right)\snorm{\vx_{\mtd, j}[\tau]}+\Nvar\vI_{N M_{\mtu}}$. For detailed proof of the noise covariance, see~\cref{app: noise covariance}. We next formulate the \gls{glrt}. 

\begin{prop}\label{thm: GLRT_perfect_CSI}
		The  estimate of target \gls{rcs} at the \gls{cpu}, denoted by $\hat{\bgamma}^{ \ttui}$, equals
  \begin{align}
\hat{\bgamma}^{ \ttui}=\bUpsilon^{ \ttui}\edit{\bpr{\sum\nolimits_{\tau=1}^{T}\ddot{{\mR}}^{H}[\tau]\bSigma_{\mts}^{-1}[\tau]\mathbf{r}_{\mtu}[\tau]}},\label{eq: gamma_opt}
     \end{align}
     where $\bUpsilon^{ \ttui}=\left(\sum\nolimits_{\tau=1}^{T}\ddot{{\mR}}^{H}[\tau]\bSigma_{\mts}^{-1}[\tau]\ddot{{\mR}}[\tau]+\bSigma_{\bgamma}^{-1}\right)^{-1}$ and $\bSigma_{\bgamma}\triangleq\Cov{\bgamma}=\diag{\sigma_{11}, \ldots, \sigma_{1\Md}, \ldots, \sigma_{\Mu\Md}}.$
     Consequently, the test statistics for target detection at the \gls{cpu}, denoted by  $\mtT$, can be evaluated as
  \begin{align}
\mtT^{ \ttui}=&\bpr{\sum\nolimits_{\tau=1}^{T}\ddot{{\mR}}^{H}[\tau]\bSigma_{\mts}^{-1}[\tau]\mathbf{r}_{\mtu}[\tau]}^{H}\hat{\bgamma}^{ \ttui}.
		\end{align}
		Hence, the hypothesis test at the \gls{cpu} can be written as $\mtT^{ \ttui}\gldec{\Hzero}{\Hone}\ln\lambda$, for a detection threshold $\lambda$.
	\end{prop} 
	\begin{proof}
	See~\cref{app: GLRT_perfect_CSI}.
	\end{proof}
\begin{obs}
    \cref{thm: GLRT_perfect_CSI} generalizes to the \gls{mle} solution presented in~\cite[Proposition~$1$]{Anubhab_Sai_Erik_ICC_2025} for centralized \gls{glrt} with \gls{rcs} being deterministic unknown.
\end{obs}
    
    Next, we observe that, unlike the case with~\gls{tdd} based multi-static sensing, the sensing noise at each \gls{ul} \gls{ap} as well as at the \gls{cpu} is colored, i.e., the covariance of effective noise is not a diagonal matrix. Consequently, we cannot express $\mtT^{ \ttui}$ in~\cref{thm: GLRT_perfect_CSI} as $\mtT_{\mtCPU}^{ \ttui}(=\sum\nolimits_{m\in\Au}\mtT_{m}^{ \ttui})$ of~\eqref{eq: GLRT_dist}~(see \cref{tables: Notation} for clarity). Hence, the distributed \gls{glrt} for target detection is \emph{not equivalent} to centralized \gls{glrt} even with perfect \gls{ul} users' \gls{csi} at the \glspl{ap} and \gls{cpu}. However, a distributed scheme incurs low front-haul \edit{overhead} compared to the centralized one. We characterize this in the following theorem.
 
\begin{thm}\label{thm: centralize_vs_distributed}
    The following relationship holds between the centralized and the distributed test metrics:
    \begin{align}
\mtT^{ \ttui}=\sum\nolimits_{m\in\Au}\mtT_{m}^{ \ttui}+\delta_{\mtT},
    \end{align}
    where the residual term $\delta_{\mtT}$ can be expressed as
    \begin{align*}
        \delta_{\mtT}=&\mathbf{t}_{2}^{H}\Delta_{1}\mathbf{t}_{2}+\mathbf{t}_{2}^{H}\bpr{\Delta_{1}+\mathbf{T}_{1}}\Delta_{2}+\Delta_{2}^{H}\bpr{\Delta_{1}+\mathbf{T}_{1}}\mathbf{t}_{2}&\notag\\&+\Delta_{2}^{H}\bpr{\Delta_{1}+\mathbf{T}_{1}}\Delta_{2},
    \end{align*}
with each variable defined as follows:
\begin{align}
&\edit{\mathbf{T}_{1}}=\bpr{\sum\nolimits_{\tau=1}^{T}\ddot{{\mR}}^{H}[\tau]\bpr{\bSigma_{\mts}^{\mtD}[\tau]}^{-1}\ddot{{\mR}}[\tau]+\bSigma_{\bgamma}^{-1}}^{-1},\nonumber\\ &\edit{\mathbf{t}_{2}}=\sum\nolimits_{\tau=1}^{T}\ddot{{\mR}}^{H}[\tau]\bpr{\bSigma_{\mts}^{\mtD}[\tau]}^{-1}\mathbf{r}_{\mtu}[\tau],\nonumber\\ &\Delta_{1}=-\mathbf{T}_{1}\bpr{\mathbf{T}_{1}^{-1}\bpr{\sum\nolimits_{\tau=1}^{T}\ddot{{\mR}}^{H}[\tau]\Delta_{\mts}[\tau]\ddot{{\mR}}[\tau]}^{-1}+\vI_{N\Mu}},\nonumber\\&\Delta_{2}=\sum\nolimits_{\tau=1}^{T}\ddot{{\mR}}^{H}[\tau]\Delta_{\mts}[\tau]\mathbf{r}_{\mtu}[\tau],\nonumber\\&\Delta_{\mts}[\tau]=-\bpr{\bSigma_{\mts}^{\mtD}[\tau]}^{-1}\bpr{{\bSigma_{\mts}^{\mtD}[\tau]}\Delta_{\vh}^{-1}+\vI_{N\Mu}}^{-1},\nonumber
\end{align}
and $\bSigma_{\mts}^{\mtD}[\tau]=\sum\nolimits_{k\in\Uu}\Eu{k}\blkdm{\vh_{mk}\vh_{mk}^{H}}+\bA[\tau]$, $\forall\tau$.
\end{thm}
\begin{proof}
    See~\cref{app: centralize_vs_distributed}.
\end{proof}
\subsubsection{Special cases}\label{sec: special cases} Using~\cref{thm: centralize_vs_distributed}, we can show $\delta_{\mtT}$ is zero for the following scenarios~(\emph{distributed \gls{glrt} attains a performance equal to the centralized}).
\begin{align}
\begin{cases} \text{Case~I}:~\vh_{k}\vh_{k}^{H}\approx\blkdm{\beta_{mk}\vI_{N}}, \forall k\in\Uu\\
       \text{Case~II}:~\text{\gls{tdd} multi-static sensing}
   \end{cases}.\label{eq: special_case}
\end{align}
The claim in~\eqref{eq: special_case} can be verified by observing that the covariance of the equivalent sensing noise at the \gls{cpu} is block-diagonal. Specifically, we can show that
\begin{align*}
    \bSigma_{\mts}[\tau]=\begin{cases}&\sum\nolimits_{k\in\Uu}\Eu{k}\blkdm{\beta_{mk}\vI_{N}}+\bA[\tau]:~\text{Case~I}\\&\bA[\tau]:~\text{Case~II}\end{cases},
\end{align*}
where $\bA[\tau]$~\eqref{eq: sensing_noise_cov} is also block-diagonal.
Consequently, 
$$\bSigma_{\mts}[\tau]=\begin{bmatrix}
\bSigma_{{\mts, 1}}[\tau] & \vZ & \ldots\\
& \ddots \\
\vZ & \ldots & \bSigma_{{\mts, \Mu}}[\tau]
\end{bmatrix},$$
implying that the computation of $\bSigma_{\mts}[\tau]$ can also be done by the $\Mu$ \glspl{ap} (block-wise) instead centrally by the \gls{cpu}; without any loss in the statistical information. Due to this \emph{de-entanglement} of the equivalent sensing noise covariance matrix,
\emph{fusing locally} obtained soft decisions~(\glspl{llr}) of the \gls{ul} \glspl{ap} at the \gls{cpu} is equivalent to finding a global \gls{llr} at the \gls{cpu} based on the full-dimensional received signal vectors. Hence, we can write
$$\mtT^{ \ttui}=\sum\nolimits_{m\in\Au}\mtT_{m}^{ \ttui}~\text{for}~\mathsf{Case~I~\&~II}.$$
This can also be rigorously shown by retracing the steps in~\cref{app: centralize_vs_distributed} with $\Delta_{\vh_{k}}=\vZ$; which yields $\delta_{\mtT}=0$ when~\eqref{eq: special_case} holds.
Finally, we note that the first approximation in \eqref{eq: special_case} is tight for large antenna density due to channel hardening~\cite{Hien_Erik_TWC_2017}.

\subsection{Bayesian Cram\'er-Rao Bound Analysis}
We now present the \gls{bcrlb} (or Van Trees bound) for the \gls{mse} for any Bayesian estimate of $\bgamma$ to benchmark~\eqref{eq: gamma_estimate_MMSE} and~\eqref{eq: gamma_opt}~\cite{Erik_Stoica_CRLB_SPL, Trees_Array_Processing}.
\begin{lem}\label{lemm: BCRLB}
    Let $\hat{\bgamma}$ be a Bayesian estimate of the \gls{rcs} $\bgamma$ at the \gls{cpu} based on the received signal vector $\mathbf{r}_{\mtu}[\tau], \tau=1,2,\ldots, T$. Then, under suitable regularity conditions, the \gls{mse} matrix, denoted by $\tt{MSE}\bbr{\hat{\bgamma}}$, is lower bounded as $\tt{MSE}\bbr{\hat{\bgamma}}\succeq \mB^{-1},$ where the \gls{bim} $\mB$ is computed as
    \begin{align}
        \mB= \bpr{\sum\limits_{\tau=1}^{T}\ddot{{\mR}}^{H}[\tau]\bSigma_{\mts}^{-1}[\tau]\ddot{{\mR}}[\tau]+\bSigma_{\bgamma}^{-1}}.
    \end{align}
\end{lem}
\begin{proof}
    See~\cref{app: BCRLB}.
\end{proof}

Next, we prove that $\hat{\bgamma}^{ \ttui}$ in~\eqref{eq: gamma_opt} is \emph{\gls{bcrlb} achieving}. To do so, we observe that $\hat{\bgamma}^{ \ttui}$ can be written as $\mB^{-1}\sum\nolimits_{\tau=1}^{T}\ddot{{\mR}}^{H}[\tau]\bSigma_{\mts}^{-1}[\tau]\mathbf{r}_{\mtu}[\tau]$ and we can evaluate its covariance as $\Cov{\hat{\bgamma}^{ \ttui}}=\bSigma_{\bgamma}\overline{\bUpsilon}\mB^{-1}$, with $\overline{\bUpsilon}\triangleq\sum\nolimits_{\tau=1}^{T}\ddot{{\mR}}^{H}[\tau]\bSigma_{\mts}^{-1}[\tau]\ddot{{\mR}}[\tau]$. Then, noting that $\tt{MSE}\bbr{\hat{\bgamma}}=\Cov{\hat{\bgamma}^{ \ttui}-\bgamma}$, and some algebra reveals $\mB\Cov{\hat{\bgamma}^{ \ttui}-\bgamma}=\vI_{\Mu\Md}$; establishing the claim. 

On the other hand, for distributed \gls{rcs} estimation,  we obtain $\hat{\bgamma}_{m}^{ \ttui}$ locally at the $m$th \gls{ul} \gls{ap}; with which we can form an overall estimate of $\bgamma$ by concatenating $\bbr{\hat{\bgamma}_{m}^{ \ttui}}$, $\forall m$, i,e., from all the \gls{ul} \glspl{ap}. To differentiate this from $\hat{\bgamma}^{ \ttui}$ in~\eqref{eq: gamma_opt}, we denote the concatenated estimate as $\hat{\bgamma}^{ \ttui}_{\tt dist.}$, with $\bsr{\hat{\bgamma}^{ \ttui}_{\tt dist.}}_{:, (m-1)\Mu:1:1:m\Mu}=\hat{\bgamma}_{m}^{ \ttui}\mbbC{\Md\times 1}$.   Now, we can characterize the \gls{mse} \gls{wrt} $\hat{\bgamma}^{ \ttui}_{\tt dist.}$ and compare it with the derived \gls{bcrlb}.
\begin{lem}\label{lem: dist_est}
It can be shown that $\tt{MSE}\bbr{\hat{\bgamma}^{ \ttui}_{\tt dist.}}=\Cov{\hat{\bgamma}^{ \ttui}_{\tt dist.}-\bgamma}=\bUpsilon^{ \ttui}_{\tt dist.}$, with $\bUpsilon^{ \ttui}_{\tt dist.}\triangleq\blkdm{\bUpsilon_{m}^{ \ttui}}$ and $\bUpsilon_{m}^{ \ttui}$ being as per~\cref{prop: GLRT1}. Also, the \gls{mse} is strictly lower-bounded as $\tr{\tt{MSE}\bbr{\hat{\bgamma}^{ \ttui}_{\tt dist.}}}>\tr{\mB^{-1}}$; because
\begin{align*}
&{\tt{MSE}}\bbr{\hat{\bgamma}^{ \ttui}_{\tt dist.}}-{\tt{MSE}}\bbr{\hat{\bgamma}^{ \ttui}}=\Delta_{{\tt MSE}}\succ\vZ_{\Mu\Md},
\end{align*}
where $\Delta_{{\tt MSE}}=\mathbf{T}_{1}\bpr{\mathbf{T}_{1}^{-1}\bpr{\sum\limits_{\tau=1}^{T}\ddot{{\mR}}^{H}[\tau]\Delta_{\mts}[\tau]\ddot{{\mR}}[\tau]}^{-1}\hspace*{-2.7mm}+\vI_{N\Mu}}$, with respective terms previously defined in~\cref{thm: centralize_vs_distributed}.
\end{lem}
\begin{proof}
The proof is similar to that of~\cref{thm: centralize_vs_distributed}.
\end{proof}

In summary,~\cref{thm: centralize_vs_distributed} and~\cref{lem: dist_est} quantify how superior the centralized scheme is over the distributed one in terms of \gls{glrt} and \gls{mse} of the estimated \gls{rcs}. However, note that the complexity of the centralized \gls{glrt} scales approximately as $\bigO{\Mu^3\Md^3}$ while for distributed it is $\bigO{\Md^3}$. 
\begin{rem}
Due to the non-diagonal sensing noise covariance matrix $\bSigma_{{\mts}}[\tau]$~\eqref{eq: sensing_noise_cov} at the \gls{cpu}, local matched filtering at the \gls{ul} \glspl{ap} does not yield sufficient statistics that can be relayed to the \gls{cpu} to estimate \gls{rcs} and compute global \gls{llr}. Hence, the centralized scheme requires $TN\Mu$ dimensional complex vectors to be sent to the \gls{cpu} via the front-haul. However, for the cases in~\cref{sec: special cases}, it can be shown that local matched filtering on $\mathbf{r}_{\mtu, m}[\tau]$~\eqref{eq: uplink_signal_mthAP}, obtained at the \glspl{ap}, is indeed the sufficient statistics based on which the \gls{cpu} can obtain a global estimate of \gls{rcs} and evaluates the \gls{glrt}; yielding the same result of a fully centralized \gls{glrt} based on the $TN\Mu$ dimensional concatenated \gls{ul} received signals. This implies, $\delta_{\mtT}=0$ in~\cref{thm: centralize_vs_distributed} and
$\Delta_{{\tt MSE}}=\vZ_{\Mu\Md}$ in~\cref{lem: dist_est}.
\end{rem}

\subsubsection{\edit{Effects of Environmental Scatterers}}\label{sec: scatter_model}
\edit{In practice, \gls{ul} \glspl{ap} also receive undesired signals due to environmental scatterers/clutters.
In principle, these signals could be measured and partially canceled beforehand if sufficient
prior information about them is available. 
To capture residual reflections, we can modify $\mathbf{r}_{\mtu, m}$ as $\mathbf{r}_{\mtu, m}^{\mathtt c}={\mathbf{r}_{\mtu, m}+\sum\nolimits_{j\in\Ad}\tilde{\mathbf{C}}_{mj}\vx_{\mtd, j}[\tau]},$
where $\tilde{\mathbf{C}}_{mj}$ is the residual clutter channel between the $m$th \gls{ul} AP and the $j$th \gls{dl} AP. Now, this only changes the equivalent noise~(which now includes also clutter and residual \gls{inai}) covariance. Thus, the results in the previous and succeeding sections follow with a minor change in the \gls{inai} covariance.
}

\section{Joint UL Signal and RCS Estimation}\label{sec: Joint}
In this section, we develop a framework for joint \gls{ul} data and target \gls{rcs} estimation. We first rewrite~\eqref{eq: uplink_signal_mthAP} as:
\begin{align}
		\mathbf{r}_{\mtu, m}[\tau]&=\overline{\vH}_{m}\bsymbol[\tau]+\ddot{{\mR}}_{m}[\tau]\bgamma_{m}+\overline{\vect{w}}_{\mtu, m}^{\mts}[\tau]\notag\\&=\underbrace{\begin{bmatrix}\overline{\vH}_{m} \quad \ddot{{\mR}}_{m}[\tau]\end{bmatrix}}_{\triangleq \bXi_{m}[\tau]}\underbrace{\begin{bmatrix}
        \bsymbol[\tau]\\
        \bgamma_{m}
\end{bmatrix}}_{\triangleq \btheta_{m}[\tau]}+\overline{\vect{w}}_{\mtu, m}^{\mts}[\tau],\label{eq: uplinkjoint_signal_mthAP}
	\end{align}
where the modified \gls{ul} channel matrix $\overline{\vH}_{m}\triangleq\vH_{m}\bEu^{1/2}$ with $\vH_{m}\triangleq\bsr{\vh_{m1}, \vh_{m2}, \ldots, \vh_{m\Ku}}\mbbC{N\times \Ku}$ and  $\bEu\triangleq\diag{\Eu{1}, \Eu{2}, \ldots, \Eu{\Ku}}$; and the modified sensing noise $\overline{\vect{w}}_{\mtu, m}^{\mts}[\tau]$ consists of only~\gls{inai} and~\gls{awgn} at the $m$th \gls{ul} \gls{ap}. Using~\cref{app: noise covariance}, we can show that $\overline{\vect{w}}_{\mtu, m}^{\mts}[\tau]\sim\cn{\vZ_{N}, \bA_{m}[\tau]}$. 
Now estimation of $\btheta_{m}[\tau], \forall \tau=1,2, \ldots, T$, can be done locally at the \gls{ul} \glspl{ap} or centrally at the \gls{cpu}, for all $m\in\Au$. We present the respective schemes next.

\subsection{Distributed Joint Estimation}
The received signal at the $m$th \gls{ap} at the end of the observation window can be expressed as
\begin{align}    \widetilde{\mathbf{r}}_{\mtu, m}=\underbrace{\begin{bmatrix}
  \overline{\vH}_{m} & \vZ & \vZ & \vZ & \ddot{{\mR}}_{m}[1]\\
  \vZ & \overline{\vH}_{m} & \vZ & \vZ & \ddot{{\mR}}_{m}[2]\\
  \vZ & \ddots & \vdots & \vZ &\ldots\\
  \vZ & \vZ &  \vZ & \overline{\vH}_{m} & \ddot{{\mR}}_{m}[T]
\end{bmatrix}}_{\triangleq \bOmegam}\underbrace{\begin{bmatrix}
\bsymbol[1] \\
\bsymbol[2]\\
\ldots \\
\bsymbol[T]\\
\bgamma_{m}
\end{bmatrix}}_{\triangleq \widetilde{\btheta}_{m}}+\widetilde{\overline{\vect{w}}}_{\mtu, m}^{\mts}.\label{eq: joint_dist_2}
\end{align}
Here, $\widetilde{\overline{\vect{w}}}_{\mtu, m}^{\mts}$ is the concatenated \gls{awgn} at the $m$th \gls{ul} \gls{ap} for $\tau=1,2, \ldots, T$.
Based on~\eqref{eq: joint_dist_2}, we have the following lemma.
\begin{lem}
    The estimate of $\widetilde{\btheta}_{m}$, denoted by $\hat{\widetilde{\btheta}}_{m}$, can be evaluated as $\hat{\widetilde{\btheta}}_{m}=\maximize{\widetilde{\btheta}_{m}}\quad\p{\widetilde{\mathbf{r}}_{\mtu, m}\vert \widetilde{\btheta}_{m}}\p{\widetilde{\btheta}_{m}}$; resulting 
\begin{multline}    \widehat{\widetilde{\btheta}}_{m}=\bpr{\bOmegam^{H}\invCov{\widetilde{\overline{\vect{w}}}_{\mtu, m}^{\mts}}\bOmegam+\invCov{\widetilde{\btheta}_{m}}}^{-1}\\\times\bOmegam^{H}\invCov{\widetilde{\overline{\vect{w}}}_{\mtu, m}^{\mts}}\widetilde{\mathbf{r}}_{\mtu, m},
\end{multline}
where  $\widetilde{\mathbf{r}}_{\mtu, m}$ is given in~\eqref{eq: joint_dist_2}, $\Cov{\widetilde{\overline{\vect{w}}}_{\mtu, m}^{\mts}}=\blkdT{\bA_{m}[\tau]}$, and  $\Cov{\widetilde{\btheta}_{m}}=\blkd{\bsr{\vI_{T}\otimes\vI_{\Ku}}, \bSigma_{\bgamma_m}}$. Finally, the estimate of the $\gls{rcs}$ $\bgamma_m$, denoted by $\hat{\bgamma}_m^{\tt J}\mbbC{\Md\times 1}$, equals the last $T\Ku+1:1:T\Ku+\Md$ entries in $\widehat{\widetilde{\btheta}}_{m}$.  
\end{lem}

\subsection{Centralized Joint Estimation}
In this case, the \gls{cpu} concatenates $\widetilde{\mathbf{r}}_{\mtu, m}$ from all \gls{ul} \glspl{ap}. Following~\eqref{eq: joint_dist_2}, the received signal in the \gls{cpu} becomes
\begin{align}   \widetilde{\mathbf{r}}_{\mtu}&\triangleq \bsr{\widetilde{\mathbf{r}}_{\mtu, 1}^{T}, \widetilde{\mathbf{r}}_{\mtu, 2}^{T},\ldots, \widetilde{\mathbf{r}}_{\mtu, \Mu}^{T} }^{T}=\bOmega\btheta_{\tt CPU}+\widetilde{\overline{\vect{w}}}_{\mtu}^{\mts},\label{eq: joint_CPU1}
\end{align}
where $\bOmega\mbbC{TN\Mu\times \bpr{T\Ku+\Md\Mu}}$ and $\btheta_{\tt CPU}\mbbC{T\Ku+\Md\Mu}$ are respectively given by
    \begin{align}        
    &\bOmega=\begin{bmatrix}
          \overline{\vH}_{1} & \vZ & \cdots & \vZ
          &\ddot{{\mR}}_{1}[1] & \cdots & \vZ \\
          \vdots & \vdots & \ddots & \vdots & \vdots & \ddots & \vdots\\ \overline{\vH}_{\Mu} & \vZ & \ldots  & \vZ & \vZ & \ldots & \ddot{{\mR}}_{\Mu}[1]\\
          \vZ & \overline{\vH}_{1} & \cdots & \vZ &\ddot{{\mR}}_{1}[2] & \cdots & \vZ\\
          \vdots & \vdots & \ddots & \vdots & \vdots & \ddots & \vdots\\
           \vZ & \overline{\vH}_{\Mu} & \cdots & \vZ & \vZ & \ldots & \ddot{{\mR}}_{\Mu}[2]
           \\
          \vdots & \vdots & \ddots & \vdots & \vdots & \ddots &  \vdots\\
          \vZ & \vZ & \cdots & \overline{\vH}_{1} & \ddot{{\mR}}_{1}[T] & \cdots & \vZ
          \\
          \vdots & \vdots & \ddots & \vdots & \vdots & \ddots & \vdots\\
          \vZ & \vZ & \cdots & \overline{\vH}_{\Mu}& \vZ & \ldots & \ddot{{\mR}}_{\Mu}[T]
   \end{bmatrix},\label{eq: omega_CPU}
   \end{align}
   and
\begin{align}
   \btheta_{\tt CPU}=\bsr{\bsymbol^{T}[1], \bsymbol^{T}[2], \ldots, \bsymbol^{T}[T], \bgamma^{T}}^{T}.\label{eq: theta_cpu}
   \end{align}
Recall that $\bgamma=\bsr{\bgamma_{1}^{T},\bgamma_{2}^{T},\ldots,\bgamma_{\Mu}^{T}}^{T}$, the concatenated \glspl{rcs} across all the \gls{ul} \glspl{ap}. Finally, $\widetilde{\overline{\vect{w}}}_{\mtu}^{\mts}$ is the concatenated~\gls{awgn} vector at the \gls{cpu}. Based on~\eqref{eq: joint_CPU1}, we have the following result.

\begin{cor}
    The joint estimate target~\gls{rcs} and \gls{ul} data symbols at the \gls{cpu}, denoted by $\hat{\btheta}_{\tt CPU}$, can be evaluated as 
    \begin{align}
    \hat{\btheta}_{\tt CPU}&=\bpr{\bOmega^{H}\invCov{\widetilde{\overline{\vect{w}}}_{\mtu}^{\mts}}\bOmega+\invCov{{\btheta}_{\tt CPU}}}^{-1}\notag\\&\qquad\qquad\quad\times\bOmega^{H}\invCov{\widetilde{\overline{\vect{w}}}_{\mtu}^{\mts}}\widetilde{\mathbf{r}}_{\mtu},
    \end{align}
    where $\Cov{\widetilde{\overline{\vect{w}}}_{\mtu}^{\mts}}=\blkd{\bA[\tau]}_{\tau=1}^{T}$ and $\Cov{{\btheta}_{\tt CPU}}=\blkd{\vI_{T}\otimes\vI_{\Ku}, \Sigma_{\bgamma}}$. Finally, the estimate of the \gls{rcs}, denoted by $\hat{\bgamma}^{\tt J}$, is the last $T\Ku+1:1:T\Ku+\Mu\Md$ entries in $\hat{\btheta}^{\tt CPU}$.
\end{cor}

Finally, we note that complexity of the joint scheme scales approximately as $\bigO{\bpr{T\Ku+\Mu}^3}$ for distributed and $\bigO{\Mu\bpr{T\Ku+\Md}^3}$ for centralized; which is $\bigO{\bpr{T\Ku}^3}$ higher compared to the framework presented in~\cref{sec: TUI}.  However, in the current formulation, we can estimate both \gls{ul} users' symbols as well as \gls{rcs} of the target.

Next, we present the \glspl{se}; capturing the effects of \gls{cli} and interference from the target reflections, i.e., \gls{tri}.

\section{Communication \& Sensing SE Analysis}\label{sec: SE and SINR}
Considering \gls{ul} communication, the equivalent noise covariance at the \gls{cpu}, denoted by $\bSigma_{\mtc}$,  can be expressed as:
\begin{align*}
\bSigma_{\mtc}&={\Ed}\blkdm{\sum\nolimits_{j\in\Ad}\sigma_{mj}\dot{\mR}_{mj}\vP_{j}{\bPi}_{j}\vP_{j}^{H}\dot{\mR}_{mj}^{H}}\\&+\Nvar\vI_{M_{\mtu}N}+\sum\nolimits_{j\in\Ad}b_{j}\diag{\zeta_{1j},\ldots,\zeta_{M_{\mtu} j}}\otimes\vI_{N},
\end{align*}
where $b_{j}\triangleq\Ed\bbr{\sum\limits_{n\in\Ud}\pi_{\mtd, n}{\snorm{\vp_{\mtd, jn}}}+\pi_{\mts}{\snorm{\vp_{\mts, j}}}}$. Here, the first term in $\bSigma_{\mtc}$ encases the effect of \gls{tri}. Next, we consider a \gls{mmse} optimal combiner at the \gls{cpu}, yielding the \gls{sinr} of the $k$th \gls{ul} user as 
\begin{align*}
		\SINRuk={\Eu{k}}\vh_{k}^{H}\bpr{\sum\nolimits_{k'\in\Uu\backslash\{k\}}\Eu{k'}\vh_{k'}\vh_{k'}^{H}+\bSigma_{\mtc}}^{-1}\vh_{k}.
\end{align*}
For \gls{dl} users, channel reciprocity~\cite{Erik_reciprocity_PIMRC} based \gls{rzf} precoder can be written as:
	\begin{align}
\vp_{\mtd,n}=\frac{\bpr{\sum\nolimits_{n'=1}^{K_{\mtd}}\vh_{n'}\vh_{n'}^{H}+\epsilon\vI_{NM_{\mtd}}}^{-1}\vh_{n}}{\norm{\bpr{\sum\nolimits_{n'=1}^{K_{\mtd}}\vh_{n'}\vh_{n'}^{H}+\epsilon\vI_{NM_{\mtd}}}^{-1}\vh_{n}}},
	\end{align}
	where $\vh_{n'}=\bsr{\vh_{1n'}^{T}, \ldots, \vh_{M_{\mtd}n'}^{T}}^{T}$ and $\epsilon>0$ is the regularizer, whose value can be optimized numerically. Next, we design two methods for sensing precoding: ``Target-centric" and ``User-centric". In the first method, we choose the sensing precoding vector $\vp_{\mts}^{\tt trg.}\mbbC{NM_\mtd}$ as the dominant right singular vector of the augmented channel matrix $\dot{\mR}$, where $$\dot{\mR}\triangleq\begin{bmatrix}
	    \dot{\mR}_{11}& \dot{\mR}_{12} & \ldots &\dot{\mR}_{1M_\mtd}\\
        \dot{\mR}_{21}& \dot{\mR}_{22} & \ldots &\dot{\mR}_{2M_\mtd}\\
        \vdots&\ddots&\vdots &\vdots\\
        \dot{\mR}_{(M_\mtu-1) 1}& \ldots & \dot{\mR}_{(M_\mtu-1) (M_\mtd-1)} &\dot{\mR}_{(M_\mtu-1) M_\mtd}\\
        \dot{\mR}_{M_\mtu 1}& {\mR}_{M_\mtu 2} & \ldots &\dot{\mR}_{M_\mtu M_\mtd}
	\end{bmatrix}.
	$$ Then, in the second method, we select the precoding vector by projecting $\vp_{\mts}^{\tt trg.}$ onto the nullspace of the concatenated channel $\bsr{\vh_{1}, \vh_{2}, \ldots, \vh_{K_{\mtd}}}$, i.e.,
	\begin{align}
\vp_{\mts}^{\tt com.}=\frac{\bpr{\vI_{NM_{\mtd}}-\vV_{\mtu}\vV_{\mtu}^{H}}\vp_{\mts}^{\tt trg.}}{\norm{\bpr{\vI_{NM_{\mtd}}-\vV_{\mtu}\vV_{\mtu}^{H}}\vp_{\mts}^{\tt trg.}}},\label{eq: p_s}
	\end{align}
	where $\vV_{\mtu}$ is a unitary matrix whose columns span the range space of the matrix $\bsr{\vh_{1}, \vh_{2}, \ldots, \vh_{K_{\mtd}}}$. Now, based on~\eqref{eq: receive_signal_ue}, the \gls{dl} \gls{sinr} can be written as:\footnote{Here, $\vx_{\mtd, j}$ is used as in~\eqref{eq: dl_tx} with the modification that $\pi_{\mtd, jn}=\pi_{\mtd, n}$ and $\pi_{\mts, j}=\pi_{\mts}$, $\forall j\in\Ad$.
    This is because, in the centralized case, the precoding vector at the $j$th \gls{dl} \gls{ap} is extracted from a centrally designed precoder at the \gls{cpu}; hence, common power control coefficients across \gls{dl} \glspl{ap} preserve the direction of the precoding vectors, which in turn retains the favorable propagation of the overall $M_{\mtd}N$ dimensional channel~\cite{making_cell_free_TWC}.}
	\begin{align*}		\SINRdn&=\frac{\pi_{\mtd, n}\Ed\abslr{\vh_{n}^{T}\vp_{\mtd, n}^{*}}^2}{\begin{pmatrix}&&\Ed\sum\nolimits_{n'\in\Ud\backslash\{n\}}\pi_{\mtd, n'}\abslr{\vh_{n}^{T}\vp_{\mtd,n'}^{*}}^2\\&&+\Ed\pi_{\mts}\abslr{\vh_{n}^{T}\vp_{\mts}^{*}}^2+\sum\nolimits_{k\in\Uu}\Eu{k}\kappa_{nk}+\Nvar\end{pmatrix}}.	
	\end{align*}  
 Here, we model the \gls{inui} channel as $\cn{0,\kappa_{nk}}$, and it is independent across all users~\cite{Bai_Sabharwal_TWC_2017, Anubhab_Chandra_TCOM_22}. Here, $\kappa_{nk}$ encapsulates the effects of path loss between the $k$th \gls{ul} and $n$th \gls{dl} user.
 
 Now, given the \gls{ul} and \gls{dl} \glspl{sinr}, the sum \gls{ul}-\gls{dl} \gls{se} becomes $\Rcomm=\sum\nolimits_{k\in\Uu}\log_{2}\bsr{1+\SINRuk}+\sum\nolimits_{n\in\Ud}\log_{2}\bsr{1+\SINRdn}.$

\subsection{Sensing \edit{\gls{scnr}} and \gls{se}}\label{sec: sensing_SCNR}

In addition to \gls{pod} associated with the \gls{glrt} or the \gls{nmse} of the estimated target \gls{rcs}, sensing \edit{\gls{scnr} and} \gls{se}~(a lower bound on the mutual information between the target's response and reflected signal) is also widely used to quantify sensing performance~\cite{MI_ISAC_TCOM};  which we present next. We later use it to underpin \gls{ap} scheduling and communication-sensing trade-off.
\begin{figure*}
\begin{align}
\edit{\SCNRs}=\frac{{\Ed}\sum\limits_{m\in\Au}\sum\limits_{j\in\Ad}\sigma_{mj}\Bigg\{\sum\limits_{n\in\Ud}\pi_{\mtd, jn}\snorm{\dot{\mR}_{mj}\vp_{\mtd, jn}}+\pi_{\mts, j}\snorm{\dot{\mR}_{mj}\vp_{\mts, j}}\Bigg\}}{\sum\limits_{k\in\Uu}\sum\limits_{m\in\Au}{\Eu{k}}\snorm{\vh_{mk}}+\Nvar NM_{\mtu}+N\sum\limits_{j\in\Ad}\sum\limits_{m\in\Au}\bbr{\zeta_{mj}+\edit{\nu_{mj}}}b_j}.\label{eq: SINR_s}
\end{align}
\hrule
\end{figure*}
\begin{lem}\label{lemma: sensing_SINR}
The sensing \gls{se} is defined as $\frac{T}{T_{\tt Ch.}}\log_{2}\bsr{1+\edit{\SCNRs}}$, where $\edit{\SCNRs}\triangleq \frac{\Elr{\snorm{\ddot{{\mR}}[\tau]\bgamma}}}{\Elr{\snorm{\vect{w}_{\mtu}^{\mts}[\tau]}}}$, with expectation being taken over the statistics of equivalent sensing noise and \gls{dl} transmitted symbols; and it evaluates to~\eqref{eq: SINR_s}. Here, $T_{\tt Ch.}$ is the total coherence duration\edit{, and $\nu_{mj}$ captures the effects of the residual clutter variance~(the power of the clutter cross-section, that also includes effects of pathloss) between the $m$th \gls{ul} AP and the $j$th \gls{dl} AP.}
\end{lem}
\begin{proof}
See~\cref{app: sensing_SINR}.
\end{proof}
   We note that in \gls{dtdd}, $\frac{T}{T_{\tt Ch.}}=1$ because the entire resource block can be used for sensing. While for \gls{tdd} only a fraction of $T_{\tt Ch.}$ can be used for sensing, based on the underlying \gls{ul}-\gls{dl} frame allocation. Thus, although in \gls{tdd} multi-static sensing, the overall sensing interference is less compared to \gls{dtdd}, the latter allows the use of an entire coherence block for sensing. This, in turn, results in superior detection, \edit{\gls{scnr}}, and, needless to say, more efficient spectral reuse. It is clear from~\eqref{eq: SINR_s}, that $\edit{\SCNRs}$ depends on the number of \gls{ul} and \gls{dl} \glspl{ap}, and also the location of the scheduled \glspl{ap}. Thus, it is pertinent to develop a scheduling algorithm that balances the communication and the sensing performance, which is presented in the next section.
\begin{rem}\label{rem: dist_SE}
\edit{\textbf{Distributed \gls{se}:} Our framework can also be extended to a distributed data processing architecture similar to~\cite{making_cell_free_TWC}; keeping in mind that, in \gls{ul}, the covariance of the interference at the $m$th \gls{ap} needs to be computed as:
\begin{align*}
\bSigma_{m, \mtc}^{\mathtt {dist.}}&={\Ed}{\sum\nolimits_{j\in\Ad}\sigma_{mj}\dot{\mR}_{mj}\vP_{j}{\bPi}_{j}\vP_{j}^{H}\dot{\mR}_{mj}^{H}}\\&+\bpr{\Nvar+\sum\nolimits_{j\in\Ad} b_{j}\bpr{\zeta_{mj}+\nu_{mj}}}\vI_{N}.
\end{align*}
With $\bSigma_{m, \mtc}^{\mathtt {dist.}}$, local \gls{mmse} combiner can be applied at each \gls{ul} \glspl{ap}, which for the $k$th stream can be expressed as:
\begin{align*}
 {\Eu{k}}\bpr{\sum\nolimits_{k'\in\Uu\backslash\{k\}}\Eu{k'}\vh_{mk'}\vh_{mk'}^{H}+\bSigma_{m, \mtc}^{\mathtt {dist.}}}^{-1}\vh_{k}.
 \end{align*}
Then, the \gls{ul} \gls{se} can be derived by a joint soft symbol decoding at the \gls{cpu} as outlined in~\cite[see equation (23)]{making_cell_free_TWC} in a straightforward manner. The \gls{dl} \gls{se} with distributed precoding follows similar to~\cite{Anubhab_Chandra_TCOM_22}, keeping in mind that the interference due to the target signal needs to be incorporated. 
}
\end{rem}
\section{AP-Scheduling}\label{sec: Scheduling}
Our analyses in the preceding sections reveal that the performance of the overall system is contingent on the underlying scheduled set of \gls{ul} and \gls{dl} \glspl{ap}. Intuitively, for communication, if the number of \gls{ul} users is more than the number of \gls{dl} users, it is pertinent to schedule more \glspl{ap} in \gls{ul} mode. However, as target detection is integrated, we need to ensure \emph{at least} one transmit and one receive \gls{ap} near the sensing area to procure measurement vectors with sufficient strength. An exhaustive search is of complexity $\bigO{2^{M}}$, and hence prohibitive. We thus propose a low-complexity approach.  

\subsubsection*{Proposed Scheme}
Let $N_{m,\mtu}$ and $N_{m,\mtd}$ be the number of \gls{ul} and \gls{dl} users within $r_{\mto}$ radius of the $m$th unscheduled \gls{ap}, respectively. Effectively, $N_{m,\mtu}$ and $N_{m,\mtd}$ represent the \gls{ul} and \gls{dl} traffic load in the vicinity~(quantified by  $r_{\mto}$) of the $m$th \gls{ap}. Now, we note that if $r_{\mto}$ is too small compared to the total area, then it can result in $N_{m,\mtu}=N_{m,\mtd}=0$. Contrarily, a large value of $r_{\mto}$ can lead to $N_{m,\mtu}$ and $N_{m,\mtd}$ being the same for several \glspl{ap}, making the scheduling process unresolvable. Thus, choosing $r_{\mto}$ is critical, and we define it as $r_{\mto}\triangleq\max\left\{\max\limits_{k\in\mathcal{U}} d_{m_{k}k}, d_{\mathsf{SNR}_{\mathsf{o}}}\right\}$,
where $m_{k}$ is the \gls{ap} index closest to the $k$th user, $d_{m_{k}k}$ is the distance between the $m_{k}$th \gls{ap} and $k$th user; and  $d_{\mathsf{SNR}_{\mathsf{o}}}$ is the distance from any user where the received \gls{snr} is at least $\mathsf{SNR}_{\mathsf{o}}$. This choice of $r_{\mathsf{o}}$ ensures that if \gls{ap}-user densities are sparse, every user has been considered to be within $r_{\mto}$ radius of at least one \gls{ap}.  However, in a dense deployment, where $\max_{k\in\mathcal{U}} d_{m_{k}k}<d_{\mathsf{SNR}_{\mathsf{o}}}$, $r_{\mto}$ ensures that local traffic loads are uniformly distributed among \glspl{ap}.
Finally, to incorporate the effect of a target on the scheduling, let $N_{\mts, \mtu}$, $N_{\mts, \mtd}$, and  $d_{m\mts}$ denote the number of \gls{ul}, \gls{dl} \glspl{ap} within $r_{\mto}$ radius of the target, and the distance between the $m$th \gls{ap} and the target, respectively.
Based on these, the pseudo-code for AP-scheduling is presented in Algorithm~\ref{algo: Scheduling}. 

\begin{algorithm}[!t]
    \caption{\gls{ul}/\gls{dl} Traffic Based AP-Scheduling}\label{algo: Scheduling}
    \SetAlgoLined
    \DontPrintSemicolon
    \KwInput{\gls{ul}/\gls{dl} user load: $N_{m,\mtu}$ and $N_{m,\mtd}, \forall m\in\A$}
    \KwOutput{Scheduled APs: $\Au, \Ad$ }   
    \KwInt{$\Au\gets\emptyset, \Ad\gets\emptyset$}
    \While  {$m=1:M$} {
             \eIf { $ N_{m,\mtu}>N_{m,\mtd}$}{
            
            $\Au=\Au\cup\{m\}$~{\% \tt Schedule in UL}
            
        }{
            $\Ad=\Ad\cup\{m\}$~{\% \tt Schedule in UL}
        }
    }
   {\% {\tt Target-centric scheduling steps} \%}\linebreak 
    \If{$N_{\mts,\mtu}=0$}
    {{Find: $i^{\mts}=  \mathtt{argmin}_{m} d_{m\mts}$\label{step: algo1}}\linebreak
  {Update: $\Au=\Au\cup\{i^{\mts}\}$ and $\Ad=\Ad \backslash \{i^{\mts}\}$}}
  \If{$N_{\mts,\mtd}=0$}{Repeat steps in~\ref{step: algo1} for \gls{dl}}\label{step: algo2}
\end{algorithm} 

Finally, note that Algorithm~\ref{algo: Scheduling} needs to be computed in the time scale proportional to slow-fading coefficients, which remain constant for several coherence intervals.

\section{Numerical Results}\label{sec: Numerical}
In this section, we provide simulation results to validate our theoretical analyses and quantify the trade-off between sensing and communication within the purview of the proposed \gls{dtdd} \gls{cf} \gls{isac} framework. We consider a $500$ square meter area, with $50\%$ of the users scheduled in the \gls{ul} mode in every slot unless specified otherwise. The bandwidth, carrier frequency, and noise figure are $20$ MHz, $1.9$ GHz, and $9$ dB, respectively~\cite{SST_EGL_SPAWC_24}. \Gls{ul} and \gls{dl} transmit powers are taken as $200$ mWatt and $1$ Watt, respectively. 
\edit{Channels for \gls{ap}-user links and the associated path-loss models for \gls{ap}-user, inter-\gls{ap}, and inter-user follow as per~\cite{cell_free_small_cells, Martin_ICASSP, Anubhab_Chandra_TCOM_22}, which pertain to $3$GPP Urban Micro-Cell models~\cite[Table B.1.2.1-1]{3gpp2010further}. The radar-related parameters follow the Swerling-$1$ model~\cite{Demir_Globecom, richards2010principles}.}  
We take $\sigma_{mj}=\sigma$, measured in dBsm. Other relevant parameters are mentioned along with the figures. We use $5000$ Monte Carlo realizations and an observation window of $100$ slots for all the experiments. Finally, the acronyms ${\tt Cent.}/{\tt Dist.}$ to indicate centralized/distributed \glspl{glrt}. 

\begin{figure*}
\centering
\begin{subfigure}{0.48\linewidth}
\centering
\includegraphics[width=0.9\textwidth]{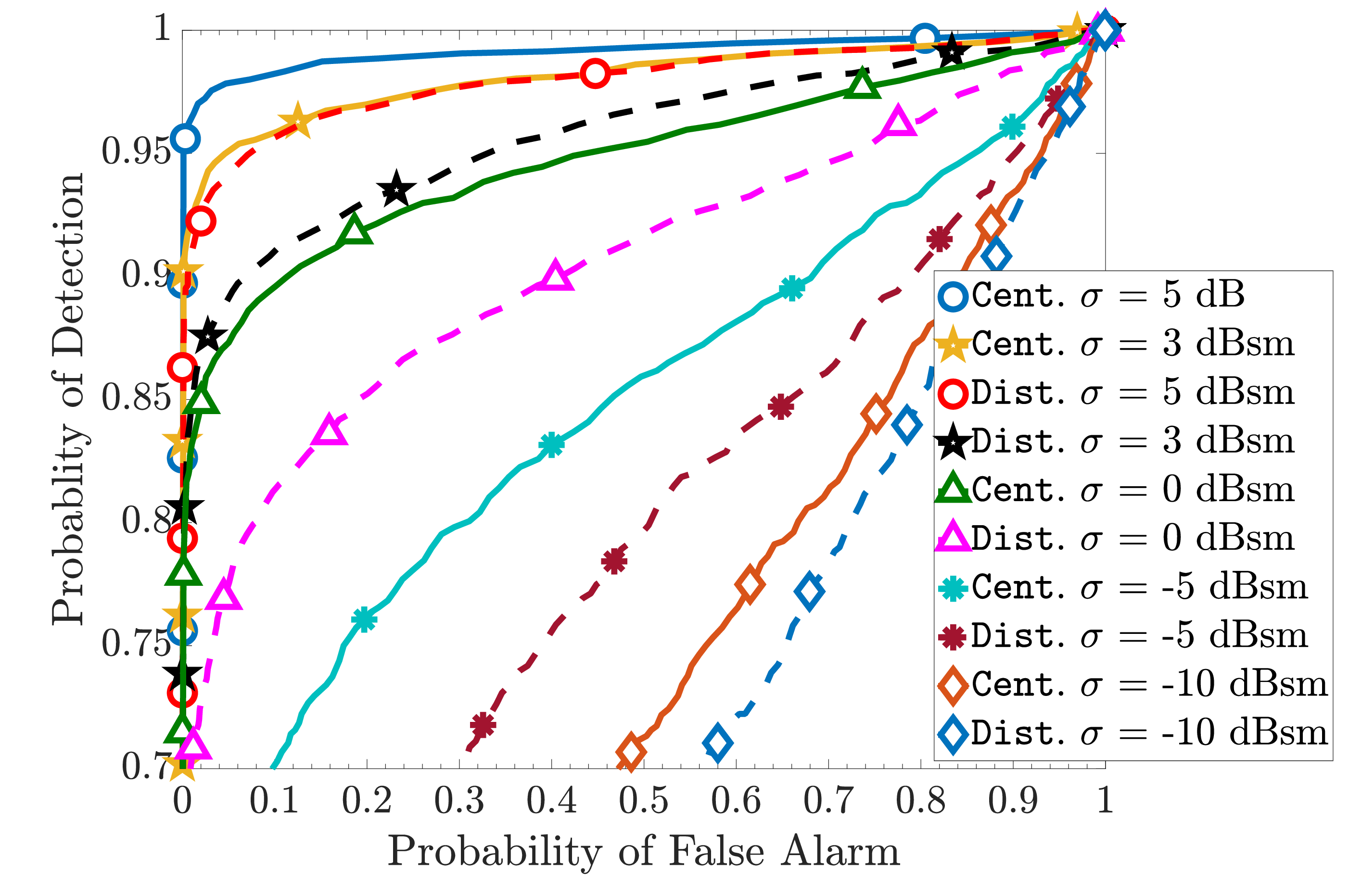}
\caption{Comparison of centralized versus distributed \glspl{glrt}. }\label{fig: DTDD_RoC}
\end{subfigure}\hfill\hfill
\begin{subfigure}{0.48\linewidth}
\centering
\includegraphics[width=0.9\textwidth]{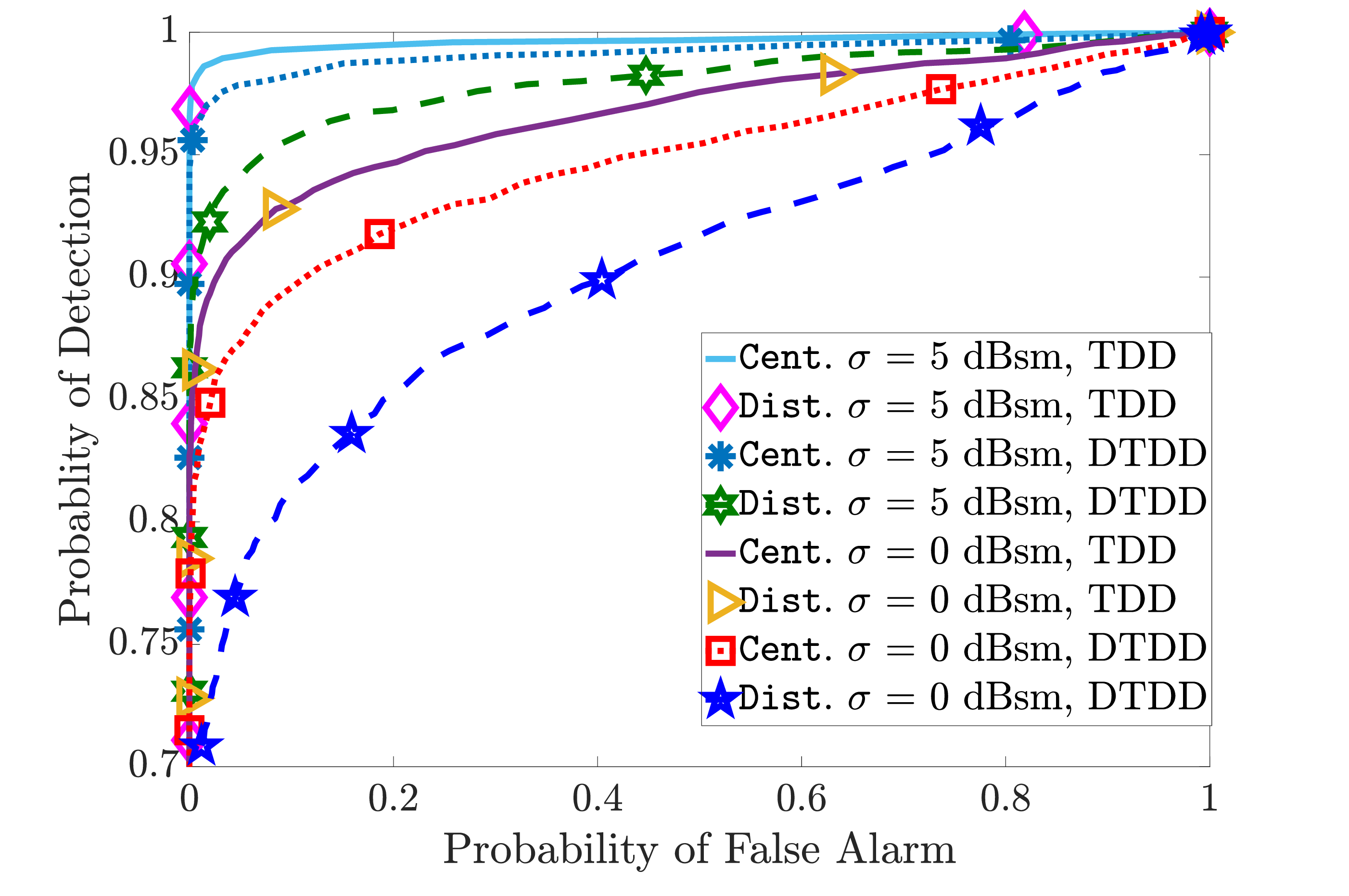}
\caption{Comparison of TDD versus DTDD.}\label{fig: TDD_DTDD_RoC}
\end{subfigure}\hfill\hfill
\caption{\Gls{roc} plots for different ranges of \gls{rcs} variance.($M=8, N=8, K=10$).}
\end{figure*}
\begin{figure}
    \centering    \includegraphics[width=.43\textwidth]{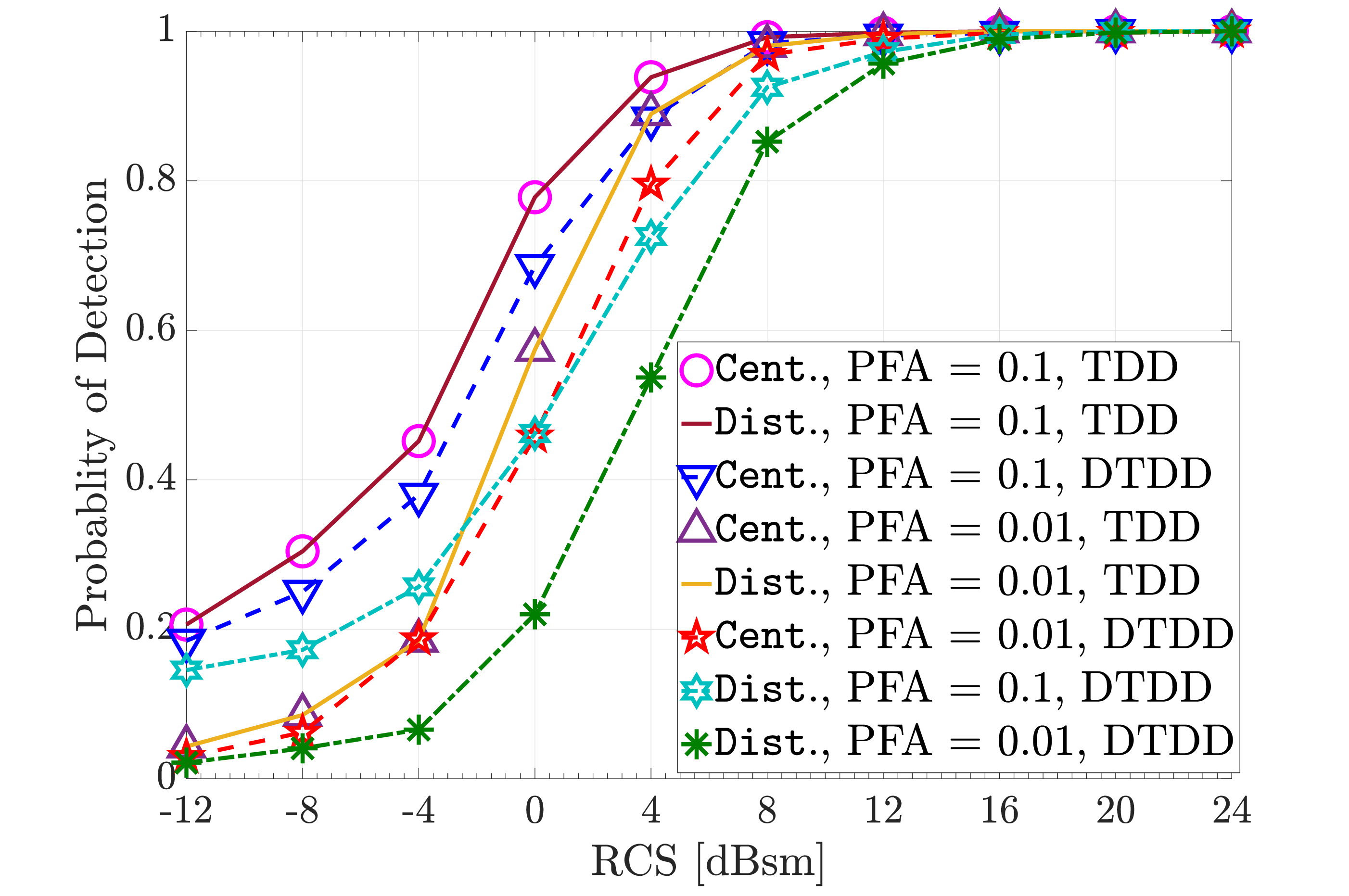}
    \caption{\Gls{pod} versus target \gls{rcs} variance considering, \edit{where we observe \gls{dtdd} procuring \gls{pod} as per \gls{tdd} as \gls{rcs} variance increases}.}
    \label{fig:P_d_P_fa}
\end{figure}

In~\cref{fig: DTDD_RoC}, we compare the performance of fully centralized versus distributed \gls{glrt} as developed in~\cref{prop: GLRT1} and~\cref{thm: GLRT_perfect_CSI}, using the \gls{roc} curves~(i.e., \gls{pod} versus \gls{pfa}). We observe that centralized \gls{glrt} uniformly outperforms distributed \gls{glrt}; which corroborates with~\cref{thm: centralize_vs_distributed}. However, when the variance of the \gls{rcs}increases, the gap in the \gls{roc} curves between centralized and distributed becomes marginal. Next, in~\cref{fig: TDD_DTDD_RoC} and~\cref{fig:P_d_P_fa}, we benchmark the performance of \gls{dtdd}-enabled multi-static sensing with traditional \gls{tdd} based system. In~\cref{fig: TDD_DTDD_RoC}, we observe that for \gls{tdd}, both centralized and distributed \gls{glrt} yield the same \gls{roc} curve, which verifies our claim in~\eqref{eq: special_case}. Further, we note that the improvement in \gls{pod} for \gls{tdd}-based system over \gls{dtdd} is due to the absence of the interference due to \gls{ul} signals in the former. However, when the variance of the \gls{rcs} increases; for instance, $5$ dBsm in~\cref{fig: TDD_DTDD_RoC}, centralized \gls{glrt} renders comparable \gls{roc} as \gls{tdd}. This validates the robustness of the \gls{glrt} to \gls{inai} and interference from the \gls{ul} users. We will later observe the marginal loss in \gls{pod} in \gls{dtdd} is a reasonable trade-off with the substantial (almost double) gain in the sum \gls{ul}-\gls{dl} \gls{se}. A similar observation can also be observed in~\cref{fig:P_d_P_fa}; where we plot \gls{pod} for a detection threshold that yields \gls{pfa} $0.1 / 0.01$. 

\begin{figure}
    \centering    
    \includegraphics[width=\linewidth]{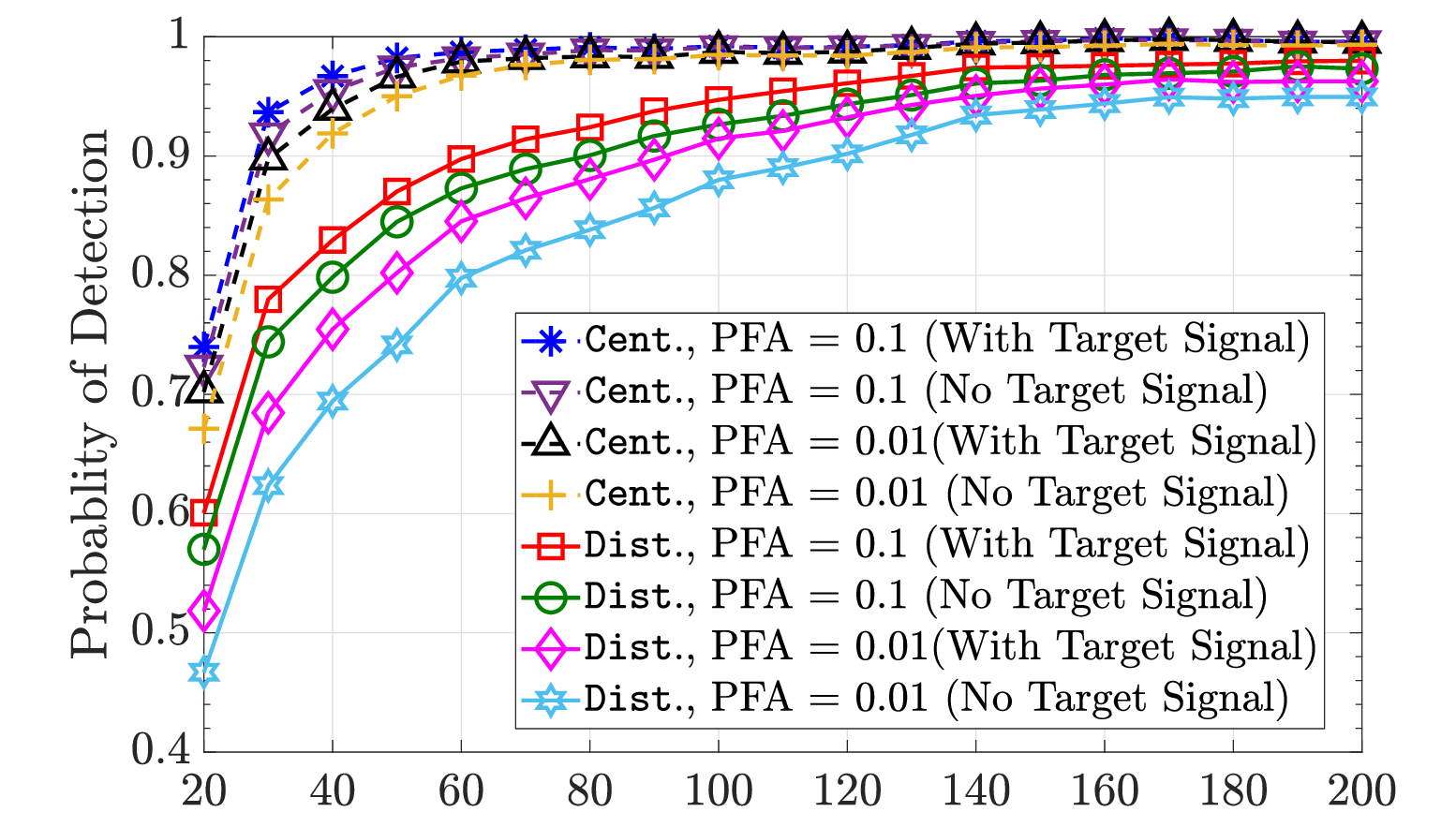}
    \caption{\Gls{pod} versus the observation window measured by the number of symbols transmitted for target detection. ($M = 10, K = 8, \sigma = 10$ dBsm.).}
    \label{fig: PoD_T_edit}
\end{figure}

\cref{fig: PoD_T_edit} illustrates the effect of sensing block duration/observation window~(i.e., $\tau$ in~\eqref{eq: uplink_signal_mthAP} and sequently) on \gls{pod}. We observe that \gls{pod} grows considerably with the increase in the observation window; and becomes stable after $\tau\approx 60$ symbols for the centralized scheme and $\tau\approx 120$ symbols for the distributed scheme. We also observe that improvement in \gls{pod} with a dedicated target signal is more pronounced for the distributed scheme than the centralized one. \edit{Now, a dense urban network yields a coherence time of $2$ ms~(theoretically permitting a mobility
of $142$ km/h) and a coherence bandwidth of $210$ kHz; the coherence interval is of length $420$ samples~\cite{Marzetta_Larsson_Yang_Ngo_2016}. Thus, $7$ hot-spot areas can be covered within a coherence block. Further, in rural areas, the coherence interval is on the order of  $15000$ samples, permitting numerous hot-spot areas.}

\begin{figure}
    \centering    
    \includegraphics[width=0.9\linewidth]{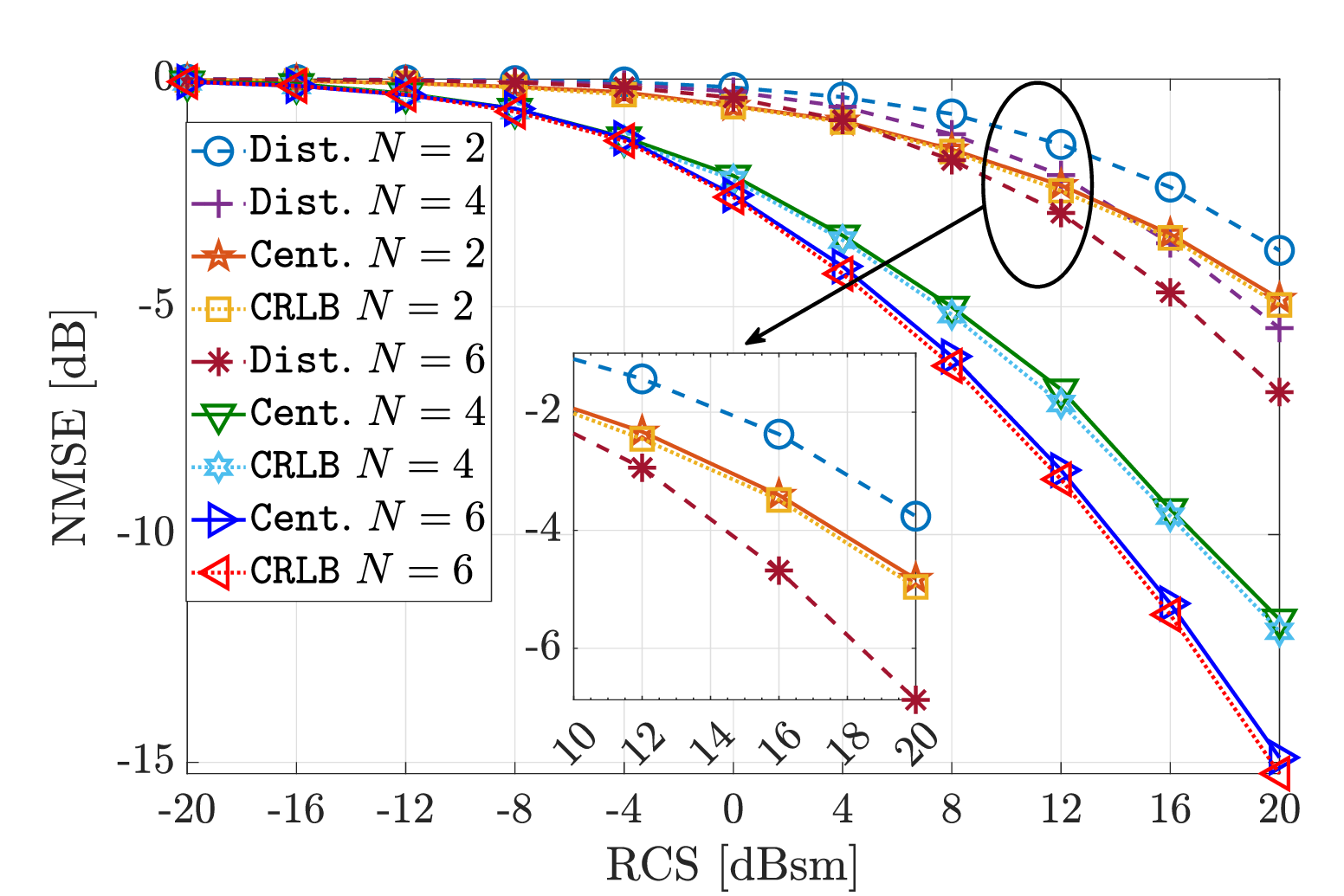}
    \caption{NMSE~(evaluated according to~\cref{prop: GLRT1} for distributed and~\cref{thm: GLRT_perfect_CSI} for the centralized) versus \gls{rcs} variance~($M=8, K=20$).}
    \label{fig: fig_NMSE_versus_RCS}
\end{figure}

\cref{fig: fig_NMSE_versus_RCS}  depicts \gls{nmse} in estimating $\bgamma$ using~\cref{prop: GLRT1} and~\cref{thm: GLRT_perfect_CSI} and compare it with the \gls{bcrlb} derived in~\cref{lemm: BCRLB}. Firstly, we observe that centralized estimation of $\bgamma$ yields an \gls{nmse} that matches with the \gls{bcrlb}. Then, we observe that for a comparatively smaller number of antennas per \gls{ap}, see the curves corresponding to $N=2$, the improvement in \gls{nmse} via the centralized scheme is insignificant over the distributed scheme. On the other hand, as per \gls{ap} antenna increases the performance improves for both centralized and distributed schemes; although the former offers substantially low \gls{nmse}. This is because, for a fixed number of users and $T$, more antennas at the \glspl{ap} procure more measurements; leading to the decrease in \gls{nmse}.

\begin{figure}
    \centering   
    \includegraphics[width=0.9\linewidth]{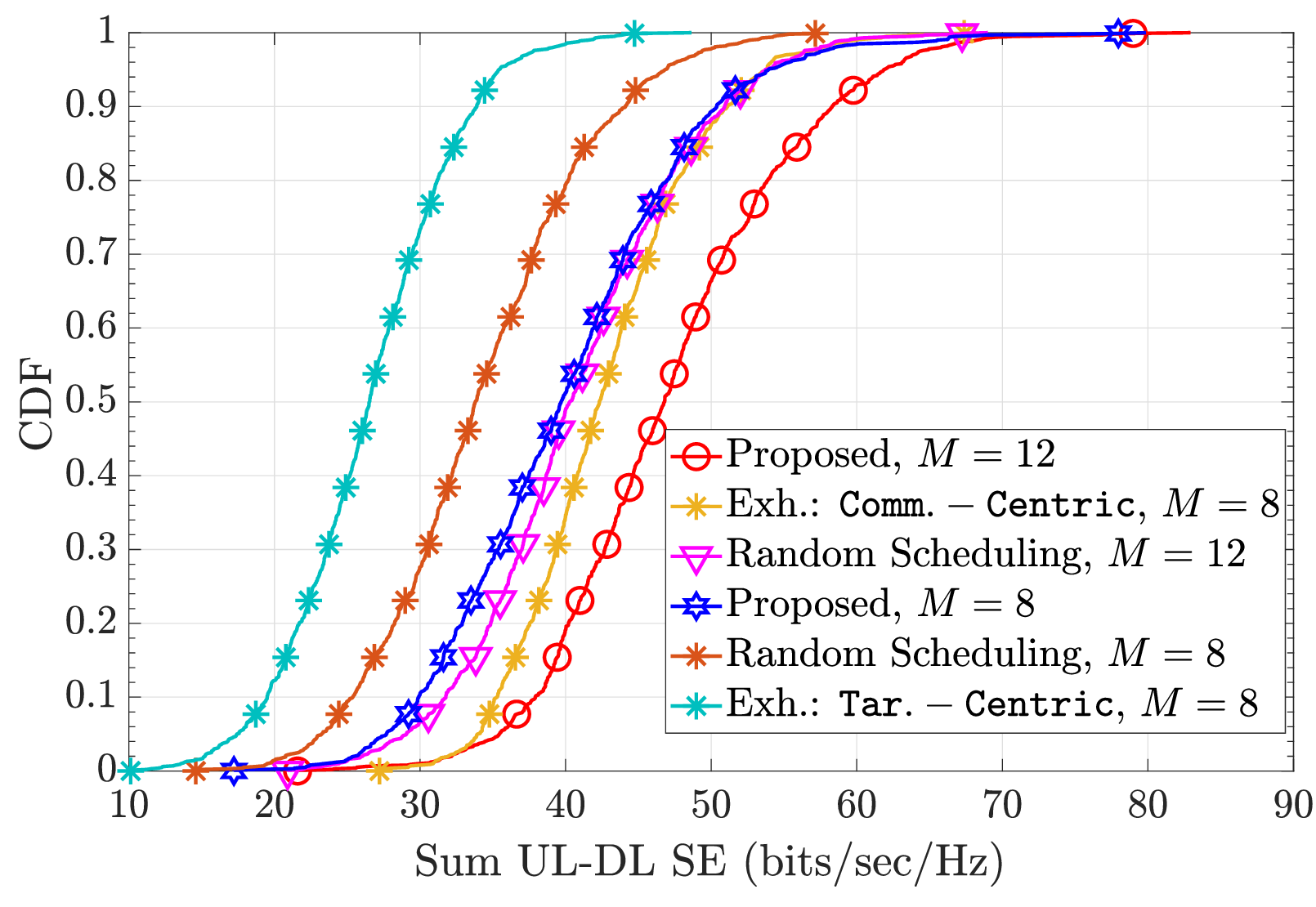}
    \caption{Comparison of the proposed scheduling scheme with exhaustive search. We consider $K = 6$ with $50\%$ of the users having \gls{ul} data demand. }
    \label{fig: scheduling2}
\end{figure}

\begin{figure}
\centering
\includegraphics[width=0.9\linewidth]{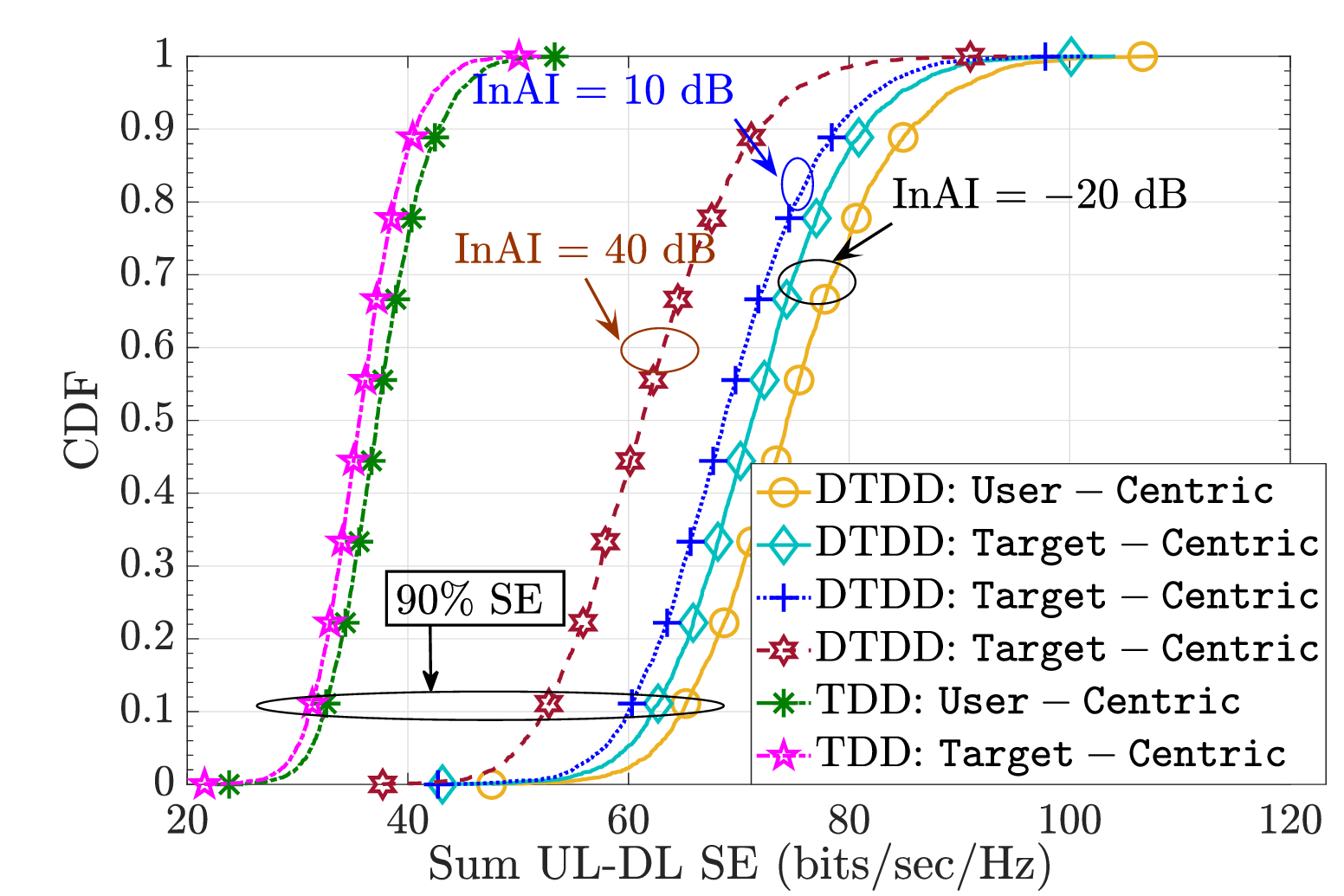}
\caption{\Gls{cdf} of sum \gls{ul}-\gls{dl} \gls{se} evaluated with \gls{dtdd} and \gls{tdd}.}\label{fig: SE1}
\end{figure}

In~\cref{fig: scheduling2}, we demonstrate the efficacy of the performance of the proposed scheduling scheme in contrast with random \gls{ap} scheduling, communication-centric exhaustive search~(Exh.: ${\tt Comm.-Centric.}$), and target-centric exhaustive search~(Exh.: ${\tt Tar.-Centric.}$) via plotting the \gls{cdf} of the sum \gls{ul}-\gls{dl} \gls{se}~(i.e., $\Rcomm$). We observe that the proposed scheme uniformly outperforms random scheduling, and the improvement becomes more pronounced for a large number of \glspl{ap}. For the communication-centric case, we search over all $2^{M}$ \gls{ul}/\gls{dl} configurations and schedule the one that procures maximum communication \gls{se}. On the other hand, for the target-centric case, we schedule the configuration that yields maximum sensing \gls{se}. We can readily observe that the target-centric scheme severely affects $\Rcomm$; while the proposed scheme keeps the balance between these two extremes; thus, it is more suitable for \gls{isac} in conjunction with \gls{dtdd}.

\cref{fig: SE1} demonstrates the performance improvement in terms of sum \gls{ul}-\gls{dl} \gls{se} that can be attained by \gls{dtdd} over \gls{tdd}-based systems, although the former incurs several \glspl{cli}. We observe that \gls{dtdd} almost doubles the $90\%$-likely sum \gls{ul}-\gls{dl} \gls{se} compared to \gls{tdd} (indicated in the figure). Further, the robustness to \gls{cli} is also a consequence of the interference aware \gls{sinr} maximizing combiner in the \gls{ul} and \gls{dl}; as we observe that even with \gls{inai} being $40$ dB above the noise floor; \gls{dtdd} procures superior performance compared to \gls{tdd}. Another critical factor is that \gls{tdd} partitions the coherence duration into \gls{ul} and \gls{dl}; thereby proportionately reducing the pre-log factor; while \gls{dtdd} uses the same time-frequency resources. Further, we observe that target-centric precoding~(indicated as $\tt Target-Centric$) incurs a marginal loss in the overall \gls{se} compared to user-centric precoding~(denoted as $\tt User-Centric$). This is because target-centric precoding in~\eqref{eq: p_s} does not nullify the interference term $\abslr{\vh_{n}^{T}\vp_{\mts}^{*}}^2$  in $\SINRdn, \forall n\in\Ud$.

\begin{figure}
    \centering    \includegraphics[width=0.9\linewidth]{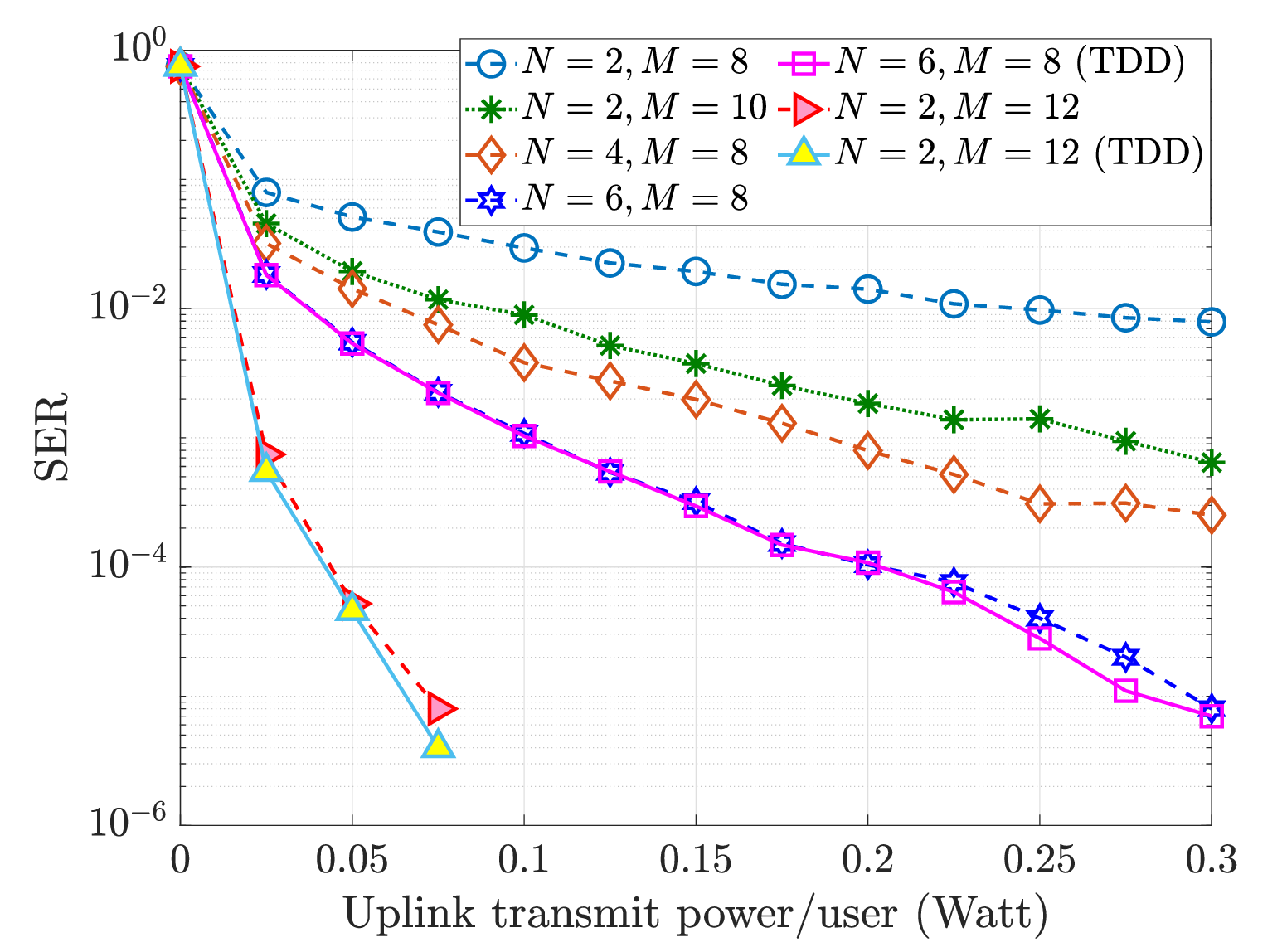}
    \caption{\gls{ul} \gls{ser} versus transmit power per \gls{ul} users. We consider total $10$ \gls{ul} users and radar \gls{rcs} variance is taken as $10$ dBsm.}
    \label{fig: ser_set1}
\end{figure}

In~\cref{fig: ser_set1}, we illustrate the behavior of \gls{ser} in \gls{ul} signal decoding obtained as a part of joint \gls{rcs} estimation and data detection \gls{wrt} per user \gls{ul} transmit power. We observe that for a given number of \glspl{ap}, increasing the number of antennas per \gls{ap} considerably reduces \gls{ser}. On the other hand, we obtain superior performance, \gls{ser} of almost $10^{-5}$ at a $0.05$ Watt transmit power, if we increase the number of \glspl{ap}, instead of per \gls{ap} antennas. This is because of the inherent link reliability obtained by distributing more \glspl{ap} in the region. Further, we observe that \gls{ser} procured for \gls{dtdd} is as good as that of a \gls{tdd} system~(i.e., no \gls{inai} and target echo during \gls{ul} data decoding phase). This shows the robustness of the developed joint decoding scheme to \gls{tri} and \gls{inai}.

\begin{figure}
\centering    \includegraphics[width=0.8\linewidth]{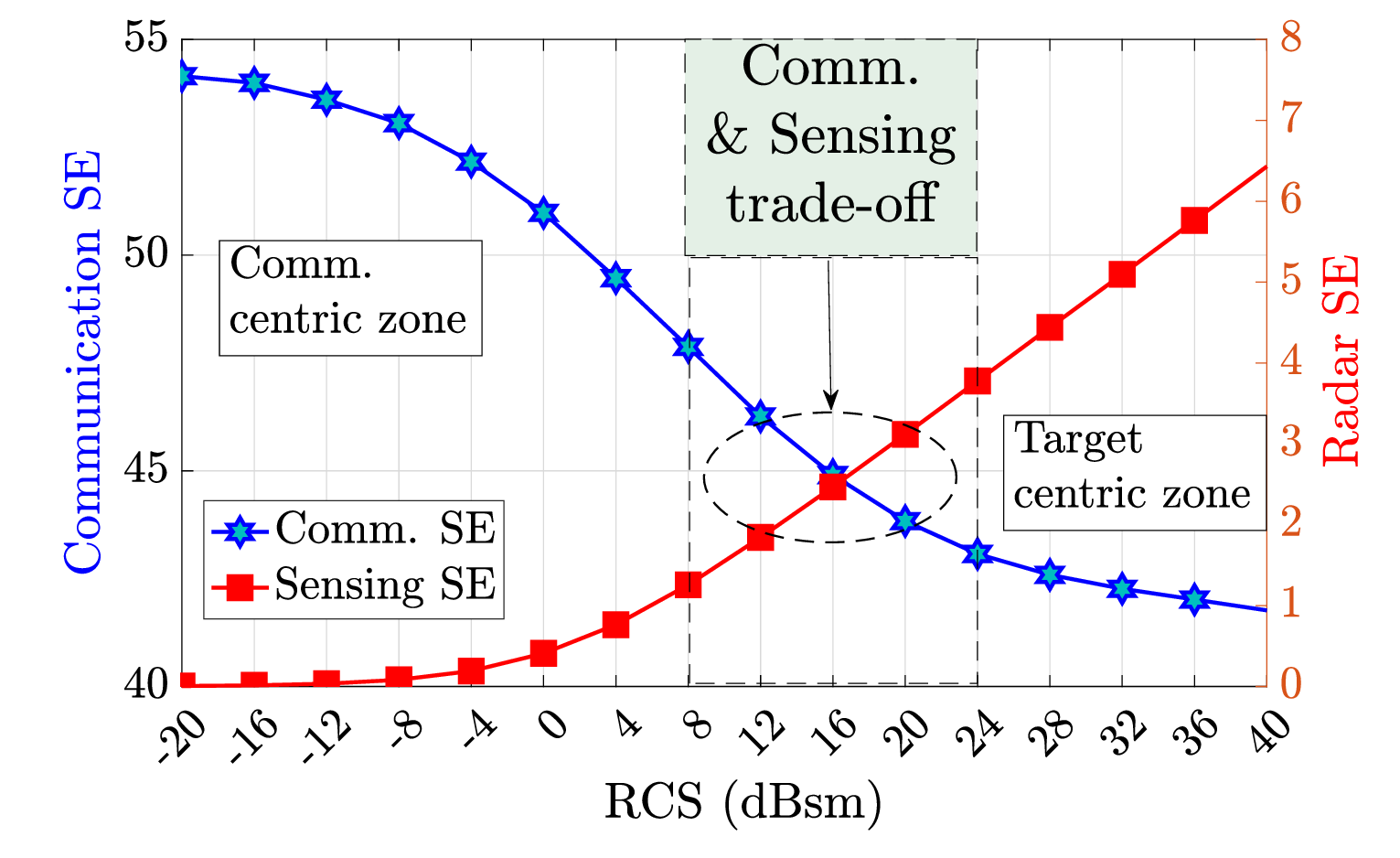}
    \caption{\edit{\Gls{isac}: Trade-off between bi-directional communication and sensing.}}
    \label{fig: fig_comm_versus_target_SE}
\end{figure}

Finally, in~\cref{fig: fig_comm_versus_target_SE}, we illustrate the trade-off between communication and sensing via plotting the respective \glspl{se} \gls{wrt} the radar \gls{rcs}. Here, we note that as the \gls{rcs} variance increases, it negatively affects the communication performance and interference strength due to the increase in received power of the target echo (\gls{tri}). On the other hand, a higher value of \gls{rcs} improves the sensing \gls{se}, as the gain in~\eqref{eq: SINR_s} enhances.

\section{Conclusion}
We proposed a novel \gls{dtdd}-enabled \gls{cf} \gls{isac} framework for target detection and bi-directional~(\gls{ul}-\gls{dl}) communication. 
We developed both centralized and distributed estimators for \gls{rcs} and \gls{glrt} for target detection, encasing the effects of \gls{inai} and interference due to the \gls{ul} users' signals. Derived \gls{bcrlb} underpinned the performances of the \gls{rcs} estimators and highlighted the conditions under which fusing the locally obtained \glspl{llr} at the \gls{cpu} is optimal. We next solved a joint target \gls{rcs} estimation and data \gls{ul} data decoding problem and showed its robustness to the \glspl{cli}. 

Next, we presented the sum \gls{ul}-\gls{dl} \gls{se} for the communication users and showed \gls{dtdd} with the proposed combiner, precoder, and scheduling; which by design are capable of mitigating the \glspl{cli}, offers two-fold improvement in $90\%$-likely \gls{ul}-\gls{dl} \gls{se} compared to \gls{tdd}; balancing the trade-off between efficient resource utilization~(i.e., concurrent \gls{ul}-\gls{dl} communication and sensing) with the additional interference~(viz. \glspl{cli} and \gls{tri}) that arise due to \gls{dtdd} and \gls{isac}. 

\edit{Detection of multiple targets with an extended target model, unknown radar channels~(parameter estimation)}, and \gls{ul}-\gls{dl} power optimization \edit{considering imperfect front-haul links} are directions for future research.

\appendix
\subsection{Derivation of equivalent sensing noise covariance.}\label{app: noise covariance}
We first note that to derive sensing noise covariance, we need to derive the power of residual \gls{inai}, which for the distributed case can be written as shown in~\eqref{eq: conditional_cov}.
\begin{figure*}
\edit{
\begin{align}
&\Elrc{\tilde{\mG}_{mj}\vx_{\mtd, j}[\tau]\vx_{\mtd, j}^{H}[\tau]\tilde{\mG}_{mj}^{H}\Bigg\vert \vx_{\mtd, j}[\tau]}{\bbr{\tilde{\mG}_{mj}}}=\Elrc{\bpr{\sum\limits_{n=1}^{N}\bsr{\tilde{\mG}_{mj}}_{:,n}\bsr{\vx_{\mtd, j}[\tau]}_{n}}\bpr{\sum\limits_{n=1}^{N}\bsr{\tilde{\mG}_{mj}}_{:,n}\bsr{\vx_{\mtd, j}[\tau]}_{n}}^{H}\Bigg\vert \vx_{\mtd, j}[\tau]}{\bbr{\tilde{\mG}_{mj}}}\notag\\&=\sum\limits_{n=1}^{N}\Elrc{\bsr{\tilde{\mG}_{mj}}_{:,n}{\bsr{\tilde{\mG}_{mj}}_{:,n}^{H}}}{\bbr{\tilde{\mG}_{mj}}}\abs{\bsr{\vx_{\mtd, j}[\tau]}_{n}}^{2}=\sum\limits_{n=1}^{N}\zeta_{mj}\abs{\bsr{\vx_{\mtd, j}[\tau]}_{n}}^{2}\vI_{N}=\zeta_{mj}\snorm{\vx_{\mtd, j}[\tau]}\vI_{N}.\label{eq: conditional_cov}
	\end{align}}
	\hrule
	\end{figure*}
In $(a)$, we model the residual \gls{inai} channel as $\bsr{\tilde{\mG}_{mj}}_{m,j}\sim\cn{0,\zeta_{mj}}$~\cite{Bai_Sabharwal_TWC_2017}. Substituting the above in \edit{ the definition $\bSigma_{{\mts, m}}[\tau]\triangleq\Elrc{\vect{w}_{\mtu, m}^{\mts}[\tau]\vect{w}_{\mtu, m}^{\mts H}[\tau] \Big\vert \bbr{\vh_{mk}, \vx_{\mtd, j}[\tau]}}{\vect{w}_{\mtu, m}^{\mts}[\tau]}$, we get the final form of $\bA[\tau]$.} 
\edit{Following similar arguments, we can show $\Elrc{\tilde{\mG}_{j}\vx_{\mtd, j}[\tau]\vx_{\mtd, j}^{H}[\tau]\tilde{\mG}_{j}^{H}\Bigg\vert\vx_{\mtd, j}[\tau]}{\bbr{\tilde{\mG}_{j}}}$ equals $\left(\diag{\zeta_{1j},\ldots,\zeta_{M_{\mtu} j}}\otimes\vI_{N}\right)\sum\nolimits_{n=1}^{N}\abs{\bsr{\vx_{\mtd, j}[\tau]}_{n}}^{2}$ completing the derivation for $\bA_{m}[\tau]$.}\hfill\qed

\subsection{Proof of~\cref{thm: GLRT_perfect_CSI}.}\label{app: GLRT_perfect_CSI}
We first note that $\p{ {\mathbf{r}}_{\mtu}[\tau]\big\vert \bgamma, \Hone}$ is $\cn{\ddot{{\mR}}[\tau]\bgamma, \bSigma_{\mts}[\tau]}$ and $\p{ {\mathbf{r}}_{\mtu}[\tau]\big\vert \Hzero}$ is $\cn{\vZ, \bSigma_{\mts}[\tau]}$. Substituting for the \glspl{pdf}, we derive the \gls{llr}, denoted by $\L$, as shown in~\eqref{eq: L_centralized_expanded}.
    \begin{figure*}
         \begin{align}
\L&=\ln\dfrac{\maximize{\bgamma}~\prod\nolimits_{\tau=1}^{T}\exp{\bbr{-\bpr{\mathbf{r}_{\mtu}[\tau]-\ddot{{\mR}}[\tau]\bgamma}^{H}\bSigma_{\mts}^{-1}[\tau]\bpr{\mathbf{r}_{\mtu}[\tau]-\ddot{{\mR}}[\tau]\bgamma}}}}{\prod\nolimits_{\tau=1}^{T}\exp{\bbr{-\mathbf{r}_{\mtu}^{H}[\tau]\bSigma_{\mts}^{-1}[\tau]\mathbf{r}_{\mtu}[\tau]}}}\notag\\&=\ln\exp\bbr{-\minimize{\bgamma}~\sum\nolimits_{\tau=1}^{T}\bbr{\bgamma^{H}\ddot{{\mR}}^{H}[\tau]\bSigma_{\mts}^{-1}[\tau]\ddot{{\mR}}[\tau]\bgamma-\bgamma^{H}\ddot{{\mR}}^{H}[\tau]\bSigma_{\mts}^{-1}[\tau]\mathbf{r}_{\mtu}[\tau]-\mathbf{r}_{\mtu}^{H}[\tau]\bSigma_{\mts}^{-1}[\tau]\ddot{{\mR}}[\tau]\bgamma}}.\label{eq: L_centralized_expanded}
         \end{align}
         \hrule
         \end{figure*}
     The optimal value of $\gamma$, denoted by $\hat{\bgamma}^{ \ttui}$, is obtained by solving $\frac{\partial }{\partial \bgamma^{H}}f(\bgamma)\vert_{\bgamma=\hat{\bgamma}^{ \ttui}}=0$, with 
     \begin{align}
f(\bgamma)&=\sum\nolimits_{\tau=1}^{T}\Bigg\{\bgamma^{H}\ddot{{\mR}}^{H}[\tau]\bSigma_{\mts}^{-1}[\tau]\ddot{{\mR}}[\tau]\bgamma-\notag\\&\bgamma^{H}\ddot{{\mR}}^{H}[\tau]\bSigma_{\mts}^{-1}[\tau]\mathbf{r}_{\mtu}[\tau]-\mathbf{r}_{\mtu}^{H}[\tau]\bSigma_{\mts}^{-1}[\tau]\ddot{{\mR}}[\tau]\bgamma\Bigg\}.
     \end{align}
     This yields~\eqref{eq: gamma_opt}.
     Then, substituting $\hat{\bgamma}^{ \ttui}$ in \gls{llr} to complete the logarithm, we obtain $\mtT^{ \ttui}$ as given in~\cref{thm: GLRT_perfect_CSI}.\hfill$\qed$

\subsection{Proof of~\cref{thm: centralize_vs_distributed}.}\label{app: centralize_vs_distributed}
We start by expanding  $\vh_{k}\vh_{k}^{H}$ as
\begin{align}
\vh_{k}\vh_{k}^{H}=\blkdm{\vh_{mk}\vh_{mk}^{H}}+\Delta_{\vh_{k}},
\end{align}
where $\Delta_{\vh_{k}}$ has the following structure:
\begin{align*}
\Delta_{\vh_{k}}=\begin{bmatrix}
        \vZ_{N\times N} && \vh_{1k}\vh_{2k}^{H} && \cdots && \vh_{1k}\vh_{\Mu k}^{H}
        \\
        \vh_{2k}\vh_{1k}^{H} && \vZ_{N\times N} && \cdots && \vh_{2k}\vh_{\Mu k}^{H}\\
        \vdots && \ddots && \ddots && \vdots
        \\
          \vh_{\Mu k}\vh_{1k}^{H} && \vh_{\Mu k}\vh_{2k}^{H} && \cdots && \vZ_{N\times N}
    \end{bmatrix}.
\end{align*}
Next, using the matrix inversion lemma,
$\bSigma_{\mts}^{-1}[\tau]$~\eqref{eq: gamma_opt} becomes
\begin{align*}
\bSigma_{\mts}^{-1}[\tau]=\bpr{\bSigma_{\mts}^{\mtD}[\tau]}^{-1}+\Delta_{\mts}[\tau],
\end{align*}
where $\bSigma_{\mts}^{\mtD}[\tau]$ is a block-diagonal matrix having the expression 
\begin{align}
\bSigma_{\mts}^{\mtD}[\tau]=\sum\nolimits_{k\in\Uu}\Eu{k}\blkdm{\vh_{mk}\vh_{mk}^{H}}+\bA[\tau],
\end{align}
and $\Delta_{\mts}[\tau]$ is a non-diagonal matrix, which is
\begin{align*}
\Delta_{\mts}[\tau]=-\bpr{\bSigma_{\mts}^{\mtD}[\tau]}^{-1}\bpr{{\bSigma_{\mts}^{\mtD}[\tau]}\Delta_{\vh}^{-1}+\vI_{N\Mu}}^{-1},
\end{align*}
with $\Delta_{\vh}=\sum\nolimits_{k\in\Uu}{\Eu{k}}\Delta_{\vh_{k}}$. Thus, we have:
\begin{align}
&\left(\sum\limits_{\tau=1}^{T}\ddot{{\mR}}^{H}[\tau]\bSigma_{\mts}^{-1}[\tau]\ddot{{\mR}}[\tau]+\bSigma_{\bgamma}^{-1}\right)^{-1}\notag\\&=\underbrace{\left(\sum\limits_{\tau=1}^{T}\ddot{{\mR}}^{H}[\tau]\bpr{\bSigma_{\mts}^{\mtD}[\tau]}^{-1}\ddot{{\mR}}[\tau]+\bSigma_{\bgamma}^{-1}\right)^{-1}}_{\triangleq \mathbf{T}_{1}}+\Delta_{1}[\tau],\label{eq: mtT1}
\end{align}
where the residual term $\Delta_{1}$ can be expressed as:
    \begin{align}
\Delta_{1}=-\mathbf{T}_{1}\bpr{\mathbf{T}_{1}^{-1}\bpr{\sum\nolimits_{\tau=1}^{T}\ddot{{\mR}}^{H}[\tau]\Delta_{\mts}[\tau]\ddot{{\mR}}[\tau]}^{-1}+\vI_{N\Mu}}.\label{eq: residual1}
    \end{align}
    Similarly, the first/third term in $\mtT^{ \ttui}$ can be written as:
    \begin{align}
{\sum\nolimits_{\tau=1}^{T}\ddot{{\mR}}^{H}[\tau]\bSigma_{\mts}^{-1}[\tau]\mathbf{r}_{\mtu}[\tau]}&=\underbrace{\sum\nolimits_{\tau=1}^{T}\ddot{{\mR}}^{H}[\tau]\bpr{\bSigma_{\mts}^{\mtD}[\tau]}^{-1}\mathbf{r}_{\mtu}[\tau]}_{\triangleq \mathbf{t}_{2}}\notag\\&\hspace{-.7cm}+\underbrace{\sum\nolimits_{\tau=1}^{T}\ddot{{\mR}}^{H}[\tau]\Delta_{\mts}[\tau]\mathbf{r}_{\mtu}[\tau]}_{\triangleq \Delta_{2}}.\label{eq: T_first_term}
    \end{align}
    Combining~\eqref{eq: T_first_term} with~\eqref{eq: residual1} we get the final result.\hfill\qed
\subsection{Proof of~\cref{lemm: BCRLB}.}\label{app: BCRLB}
For \gls{bcrlb}, the \gls{bim} $\mB$ can be expressed as $\mB=\mB_{\tt D}+\mB_{\tt P}$, where $\mB_{\tt D}$ denotes the information matrix due to the data and $\mB_{\tt P}$ indicates the information matrix due to prior~\cite[Chapter $8$]{Trees_Array_Processing}. Next, we can evaluate the $\bpr{i,j}$th element of $\mB_{\tt D}$ and $\mB_{\tt P}$ as $\bsr{\mB_{\tt D}}_{i,j}=-\Elr{\frac{\partial^2}{\partial \bsr{\bgamma}_{i}\partial \bsr{\bgamma}_{j}}\ln\prod\nolimits_{\tau=1}^{T}\p{\mathbf{r}_{\mtu}[\tau]\lvert \bgamma}}$ and $\bsr{\mB_{\tt P}}_{i,j}=-\Elr{\frac{\partial^2}{\partial \bsr{\bgamma}_{i}\partial \bsr{\bgamma}_{j}}\ln\p{\bgamma}}$, with $i,j=1, 2,\ldots, \Mu\Md$ and $\bsr{\bgamma}_{i}$ being the $i$th element of $\bgamma$.
    Now, recall that our observation vector $\mathbf{r}_{\mtu}$ has the conditional \gls{pdf} $\p{\mathbf{r}_{\mtu}[\tau]\lvert \bgamma}=\cn{\ddot{{\mR}}[\tau]\bgamma,\bSigma_{\mts}[\tau]}$, for $\tau=1,2,\ldots, T$; and the prior \gls{pdf} $\p{\bgamma}=\cn{\vZ_{\Mu\Md},\bSigma_{\bgamma}}$. Algebraic evaluation of the above information metrics yields the final result.\hfill \qed

\subsection{Proof of~\cref{lemma: sensing_SINR}.}\label{app: sensing_SINR}
We note that the expectations in $\edit{\SCNRs}$ are taken over the statistics of~\gls{rcs} and also the transmitted \gls{dl} and sensing data $\vx_{\mtd,j}[\tau]$. Thus, the average \edit{sensing \gls{scnr}} is independent of the slot index $\tau$. Next, we can show that
\begin{align}
    &\Elr{\snorm{\ddot{{\mR}}[\tau]\bgamma}}=\tr{\Elr{{\ddot{{\mR}}[\tau]\Elr{\bgamma\bgamma^{H}}\ddot{{\mR}}^{H}[\tau]}}}\notag\\
    &\stackrel{(a)}{=}\tr{\blkdm{\Elr{\ddot{{\mR}}_{m}[\tau]\bsr{\bSigma_{\bgamma}}_{m}\ddot{{\mR}}_{m}^{H}[\tau]}}}
    \notag\\
    &\stackrel{(b)}{=}{\Ed}\tr{\blkdm{\sum\nolimits_{j\in\Ad}\sigma_{mj}\dot{\mR}_{mj}\vP_{j}{\bPi}_{j}\vP_{j}^{H}\dot{\mR}_{mj}^{H}}}\notag\\&={\Ed}\sum\nolimits_{m\in\Mu}\sum\nolimits_{j\in\Md}\sigma_{mj}\tr{\dot{\mR}_{mj}\vP_{j}{\bPi}_{j}\vP_{j}^{H}\dot{\mR}_{mj}^{H}}
    \notag\\&={\Ed}\sum\nolimits_{m\in\Au}\sum\nolimits_{j\in\Ad}\sigma_{mj}\Bigg\{\sum\limits_{n\in\Ud}\pi_{\mtd, jn}\snorm{\dot{\mR}_{mj}\vp_{\mtd, jn}}\notag\\&\qquad\qquad\qquad\qquad+\pi_{\mts, j}\snorm{\dot{\mR}_{mj}\vp_{\mts, j}}\Bigg\},\label{eq: sensing_SINR_1}
\end{align}
where in $(a)$, $\bsr{\bSigma_{\bgamma}}_{m}$ denotes the $\bbr{(m-1)\Md+1:1:m\Md}$th block of entries of the matrix ${\bSigma_{\bgamma}}$. Then, $(b)$ follows using
$\Elr{\vx_{\mtd, j}[\tau]\vx_{\mtd,  j}^{H}[\tau]}={\Ed}\vP_{j}{\bPi}_{j}^{\frac{1}{2}}\Elr{\vs_{\mtd}[\tau]\vs_{\mtd}^{H}[\tau]}\bPi_{j}^{\frac{1}{2}}\vP_{j}^{H},$
with $\Elr{\vs_{\mtd}[\tau]\vs_{\mtd}^{H}[\tau]}=\vI_{\Kd}$. 
Next, the denominator of $\edit{\SCNRs}$ can be evaluated as
$\Elr{\snorm{\vect{w}_{\mtu}^{\mts}[\tau]}}
=\sum\nolimits_{k\in\Uu}\sum\nolimits_{m\in\Au}{\Eu{k}}\snorm{\vh_{mk}}
+\Nvar NM_{\mtu}+N\sum\nolimits_{j\in\Ad}\sum\nolimits_{m\in\Au}\bbr{\zeta_{mj}+\edit{\nu_{mj}}}b_j$\edit{, with $\nu_{mj}$ capturing the effects of clutter~\cite{Jeong_TVT}.} This completes the proof.
\hfill \qed

\ifCLASSOPTIONcaptionsoff
\newpage
\fi
\bibliographystyle{IEEEtran.bst}
\typeout{}
\bibliography{bibfile_ISAC_CF}

\end{document}